\documentclass[11pt]{article}
\usepackage[margin=1in]{geometry}
\usepackage{amsmath, amssymb, amsthm}
\usepackage{url}
\usepackage{xcolor,soul}
\usepackage{hyperref}
\usepackage[numbers,sort&compress]{natbib}
\usepackage{adjustbox}
\usepackage{bm}
\DeclareMathOperator{\R}{\mathbb{R}}
\DeclareMathOperator{\PLRV}{\mathrm{PLRV}}
\usepackage{amsmath, amsthm}
\newtheorem{remark}{Remark}
\newtheorem{prop}{Proposition}
\usepackage{amsthm}
\usepackage{amsmath}
\DeclareMathOperator{\Tr}{Tr}
\usepackage{thmtools} 
\newtheorem{corollary}{Corollary}
\theoremstyle{plain} 
\newtheorem{theorem}{Theorem}[section]
 
\theoremstyle{definition} 
\newtheorem{definition}{Definition}

\title{Bridging Differential Privacy and Random Triangles}
\author{
Tianxi Ji\\
Texas Tech University\\
Department of Computer Science\\
\texttt{tiji@ttu.edu}}
\date{}

\begin{document}
\maketitle
\vspace{-1.5em}
\begin{figure}[htb]
  \centering
        \adjustbox{trim=28pt 0pt 20pt  0pt,clip}
        {\includegraphics[width=0.95\linewidth]{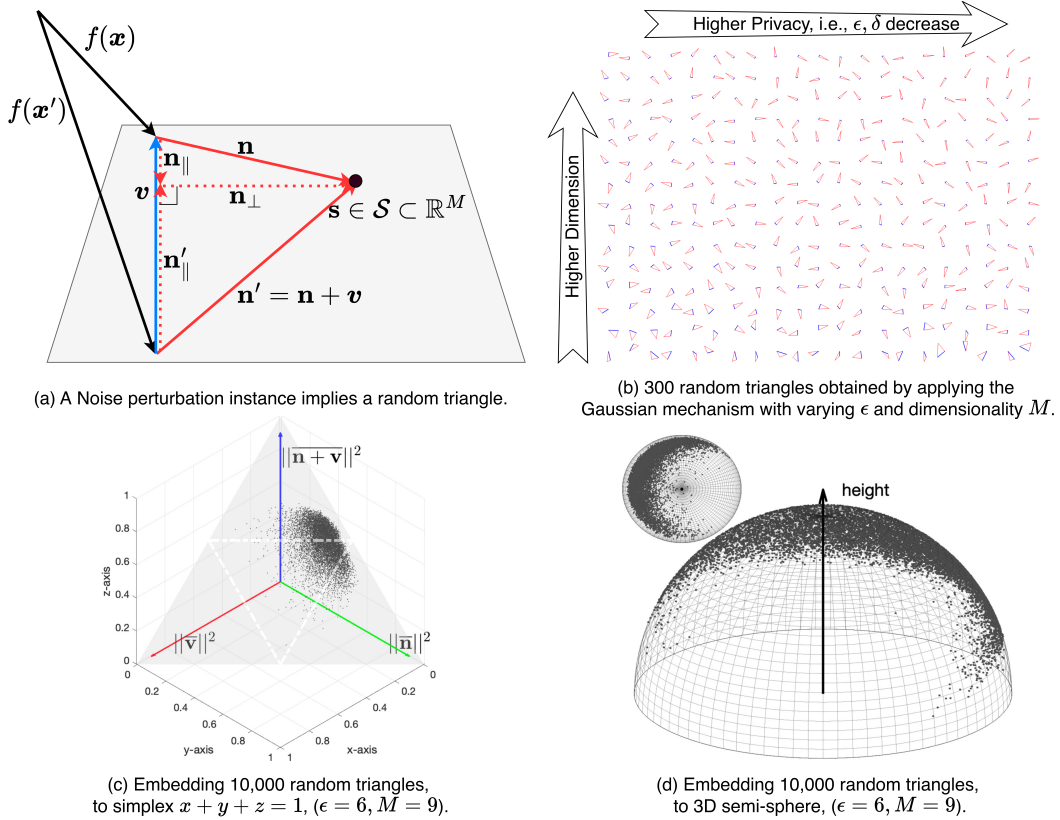}}
\vspace{-1em}
    \caption{Noise perturbation implies random triangles.}
    \label{fig:geo-dp}
    \vspace{-1em}
\end{figure}

\begin{abstract}
The classical analysis of the Gaussian mechanism in differential privacy reduces privacy loss for a
pair of neighboring datasets to a scalar random variable. While this scalar characterization
is sufficient for   privacy accounting, each perturbation instance
also induces a high-dimensional random triangle formed by the sensitivity
vector and the two corresponding noise vectors. In this work, we develop two
complementary geometric representations of these random triangles.

The first representation maps the normalized squared edge lengths to a
 simplex. We derive its exact joint density, characterize its
elliptical support, and reconstruct the classical privacy loss random variable
from the simplex coordinates. The
second representation maps the spectral shape of each normalized triangle to
a hemisphere. We
derive the corresponding density and coordinate mappings, recover the same
privacy loss, and characterize an equatorial drift together with band
concentration as the dimension increases. These results provide two exact
geometric coordinate systems that complement the scalar privacy loss   and 
connect differential privacy with the probabilistic analysis of random
shapes.
\end{abstract}

\section{Introduction}
\label{sec:intro}

In high-dimensional data analysis and machine learning, the Gaussian
mechanism is a standard approach for achieving $(\epsilon,\delta)$-differential
privacy (DP)~\cite{dwork2014algorithmic}. Given an $M$-dimensional query/computation result  $f(\boldsymbol{x})\in\mathbb{R}^M$, it
releases $\mathcal{M}(\boldsymbol{x})=f(\boldsymbol{x})+\mathbf{n}$ with $\mathbf{n}\sim\mathcal{N}(\mathbf{0},\sigma^2 \mathbf{I}_M)$, 
where the scale of the noise $\sigma$ is calibrated to hide the sensitive difference vector $\bm{v}  \triangleq f(\bm{x})-f(\bm{x'})$ caused by any arbitrary pair of neighboring dataset, $\bm{x}$ and $\bm{x}'$, that differ only by one record.

The privacy guarantee of the Gaussian mechanism is determined by the distribution of the privacy loss random variable (PLRV), which simplifies algebraically to 
\begin{equation}
    \text{PLRV}
    =
    \log
    \frac{
        \exp\left(
            -\frac{||\mathcal{M}(\boldsymbol{x})-f(\boldsymbol{x})||_2^2}{2\sigma^2}
        \right)
    }{
        \exp\left(
            -\frac{||\mathcal{M}(\boldsymbol{x})-f(\boldsymbol{x}')||_2^2}{2\sigma^2}
        \right)
    }
    =
    \frac{
        ||\mathbf{n}+\boldsymbol{v}||_2^2
        -
        ||\mathbf{n}||_2^2
    }{2\sigma^2}
   =
    \frac{\langle \mathbf{n},\boldsymbol{v}\rangle}{\sigma^2}
    +
    \frac{||\boldsymbol{v}||_2^2}{2\sigma^2}.
    \label{eq:v-plrv}
\end{equation}
Since $\langle \mathbf{n},\boldsymbol{v}\rangle
    \sim
    \mathcal{N}\left(
        0,
        \sigma^2||\boldsymbol{v}||_2^2
    \right)$, we have $\text{PLRV}
    \sim
    \mathcal{N}\left(
        \frac{||\boldsymbol{v}||_2^2}{2\sigma^2},
        \frac{||\boldsymbol{v}||_2^2}{\sigma^2}
    \right)$. When taking the worst case over all neighboring datasets and defining the sensitivity of the computation, $   \Delta_2 f
    =
    \sup_{\boldsymbol{x}\simeq\boldsymbol{x}'}
    ||f(\boldsymbol{x})-f(\boldsymbol{x}')||_2=\sup_{\boldsymbol{x}\simeq\boldsymbol{x}'}
    ||\boldsymbol{v}||_2$, 
the induced PLRV is a univariate Gaussian, i.e., 
\begin{equation}\label{eq:plrv_reduction}
\text{PLRV}\sim \mathcal{N}\left(\frac{(\Delta_2f)^2}{2\sigma^2},\frac{(\Delta_2f)^2}{\sigma^2}\right)\text{~\cite{balle2018improving}}. 
\end{equation}

Clearly,  even if $f(\bm{x})$ is high-dimensional, PLRV 
only depends \textit{exclusively} on the inner product $\langle \mathbf{n}, \bm{v} \rangle$, i.e., the component of the noise \textit{parallel} to the sensitivity vector $\bm{v}$. The vast majority of the noise energy, stored in the $M-1$ dimensions \textit{orthogonal} to $\bm{v}$ (i.e., $\mathbf{n}_{\perp}$ visualized in Figure~\ref{fig:geo-dp} (a)), cancels out entirely and vanishes from the privacy analysis. The cancellation of $\mathbf{n}_{\perp}$ is due to the spherically
symmetric property of multivariate Gaussian distribution~\cite{ali1980characterization}.


In this work, we ask a \textbf{geometric question}: what random shapes are generated by
the Gaussian mechanism before the high-dimensional perturbation is reduced
to the scalar privacy loss?     This question  does not indicate any limitation or contradiction in the
classical privacy analysis. In other words, the scalar PLRV is sufficient for characterizing
the pairwise privacy guarantee of the Gaussian mechanism. Nevertheless, it does not retain the full geometry of the
underlying perturbation when obscuring $f(\boldsymbol{x})$ and $f(\boldsymbol{x}')$.  To answer the question, we   shift the perspective from 1D algebraic projections to a geometric landscape, i.e.,   instances of  random triangles  (as   illustrated in Figure~\ref{fig:geo-dp} (a)). Geometrically, for any randomized output $\mathbf{s}\in\mathcal{S}\subset\R^M$, the noise vectors $\mathbf{n}$ and $\mathbf{n}'$ (used to perturb $f(\bm{x})$ and $f(\bm{x}')$) inherently form a closed triangle with the sensitivity vector $\bm{v}  \triangleq f(\bm{x})-f(\bm{x'})$~\cite{ji-r1smg,liu2026privacy} (see details in Section~\ref{sec:noise-and-triangles}).

Different from ~\cite{ji-r1smg,liu2026privacy}, which leverage the geometry of random triangles to design novel perturbation noises for improved utility, we shift the focus back to the foundational   Gaussian mechanism. We propose two   \textbf{geometric views} to  characterize the statistical   properties of the random triangles naturally induced by classic isotropic Gaussian perturbations:

$\bullet$ \textbf{View 1: Macroscopic (Metric) Perspective.} We first study  the ``macro" picture, i.e.,   the overall shapes of the random triangles formed by various instances of a Gaussian mechanism. We map instances of the high-dimensional random triangle onto a 3D probability simplex (i.e., the plane defined by $x+y+z=1$, see Figure~\ref{fig:geo-dp}(c)). By deriving the exact joint Probability Density Function (PDF) on this simplex, we express the classical   PLRV as an explicit scalar
function of the simplex coordinates. We also show that the high-dimensional geometry is governed by two simultaneous
effects. First, the orthogonal noise   restricts the simplex point clouds to a
rotated-ellipse admissible in the \((x,y)\)-coordinates. 
 Second, as the dimension grows, the dominant mass of the PDF is
pushed along this region toward the degenerate noise-dominated configuration
\((x,y,z)=(0,1/2,1/2)\).

$\bullet$ \textbf{View 2: Microscopic (Spectral) Perspective.} Next, we zoom into the internal structure of these triangles. By employing Singular Value Decomposition (SVD), we embed the high-dimensional random triangles onto a 3D hemisphere, see Figure~\ref{fig:geo-dp}(d). We establish the
bijection between the triangle Gram matrix and the hemisphere coordinates,
derive the exact joint density of longitude and latitude, and express the
Gaussian PLRV in these coordinates. The second view provides a complementary description of the
high-dimensional limit. As $M$ increases, the hemisphere distribution
moves toward lower latitudes at rate $M^{-1/2}$ while becoming concentrated
within an increasingly thin band near the equator. 

\textbf{Contributions.} This work, for the first time,  provides a fundamental and comprehensive geometric characterizations of the classical Gaussian mechanism in high dimensions.  
Each provides an exact geometric coordinate system
    in which the mechanism's shape distribution can be   studied.


\section{Related Work}\label{sec:relatedwork}

Random triangles have been studied more broadly within statistical shape
analysis, where geometric configurations are represented after removing
nuisance transformations such as translation, rotation, and scale~\cite{small2012statistical,kendall1989survey}.
Normalized Gram matrices and singular-value coordinates provide
rotation-invariant descriptions of shape, and the shape space of triangles
admits natural spherical or hemispherical representations
\cite{edelman2015random}. 

In differential privacy, the random-triangle view of noise perturbation was introduced in~\cite{ji-r1smg} and subsequently extended in~\cite{liu2026privacy}. Both works leverage measure concentration of random geometric objects to design new perturbation mechanisms, while the latter relaxes the strong independence assumption on the noise imposed in the former.

In this paper, we keep the classical isotropic Gaussian mechanism fixed and
study the probability law of the triangle shapes it induces. Our visualization
and geometric processing of random triangles build on standard techniques
from~\cite{edelman2015random}. The probabilistic models, however, are
different: \cite{edelman2015random} considers triangles whose three vertices
are independently sampled from Gaussian distributions, whereas our triangles
are generated by a single Gaussian noise vector together with a fixed
sensitivity displacement. 
Consequently, the remaining geometric components
are determined by the relation 
    $\mathbf n'=\mathbf n+\boldsymbol{v}$, 
rather than being sampled independently.

\section{Preliminaries}\label{sec:preliminary}


\begin{definition}
\label{def_dp}
A randomized mechanism $\mathcal{M}$ satisfies $(\epsilon,\delta)$-DP if for any two neighboring datasets, $\bm{x},\bm{x}'\in\mathbb{N}^{|\mathcal{X}|}$ that differ by only one data record, 
$\epsilon>0$ and $0< \delta <1$, it satisfies $\Pr[\mathcal{M}(\bm{x})\in \mathcal{S}] \leq e^{\epsilon}\Pr[\mathcal{M}(\bm{x}')\in \mathcal{S}]+\delta$, 
where $\mathcal{S}$ denotes the output set. 
\end{definition}


\begin{theorem}\label{thm:gaussian_privacyguarantee} [Gaussian Mechanism and Privacy  Guarantee~\cite{dwork2014algorithmic}] 
Let $f:\mathbb{N}^{|\mathcal{X}|}\rightarrow \mathbb{R}^M$ be an arbitrary $M$-dimensional function. 
The Gaussian mechanism with parameter $\sigma$ adds random Gaussian noise following $\mathcal{N}(0,\sigma^2)$ to each of the $M$ components of the output, i.e., $\mathbf{s} = f(\bm{x}) + \mathbf{n}$ and $\mathbf{n}\sim\mathcal{N}(\mathbf{0},\sigma^2\mathbf{I})$. 
Let $\epsilon\in(0,1)$  and  $\delta\in(0,1)$ be arbitrary. The   Gaussian Mechanism achieves ($\epsilon,\delta$)-DP, if $\sigma\geq  \frac{c\Delta_2f}{\epsilon}$ and $c^2 > 2\ln(\frac{1.25}{\delta})$.
\end{theorem}


\begin{definition}[
PLRV~\cite{bun2016concentrated}]\label{def:plrv}
Let  $\bm{x}, \bm{x}' \in \mathbb{N}^{|\mathcal{X}|}$ be two neighboring datasets, $\mathcal{M}: \mathbb{N}^{|\mathcal{X}|} \to \mathcal{S}$ be a randomized mechanism, and   $P$ and $Q$ be the   distributions of $\mathcal{M}(\bm{x})$ and $\mathcal{M}(\bm{x}')$, respectively.   PLRV associated with   $\mathcal{M}$ and a randomized output $\mathbf{s}$ is defined as $\PLRV_{\mathcal{M}}(s) = \ln \left( \frac{\mathrm{d} P(\mathcal{M}(\bm{x})=\mathbf{s})}{\mathrm{d} Q(\mathcal{M}(\bm{x}')=\mathbf{s})} \right)$, 
where $\frac{\mathrm{d} P(\mathcal{M}(\bm{x})=\mathbf{s})}{\mathrm{d} Q(\mathcal{M}(\bm{x}')=\mathbf{s})}$ denotes the Radon-Nikodym derivative of $P$ w.r.t. $Q$ evaluated at $s \in \mathcal{S}$.
\end{definition}

\section{The Proposed Geometric Views}\label{sec:main-tech}

In this section, we first briefly review recent developments linking DP to random triangles~\cite{ji-r1smg,liu2026privacy}. 
Next, we present the proposed geometric views, which provide a novel interpretation of DP. 

\subsection{Noise Perturbation and   Random Triangles}
\label{sec:noise-and-triangles}



 
 To ensure 
 DP when releasing a  query $f(\bm{x})\in\mathbb{R}^M$ on a sensitive dataset $\bm{x}$, noise perturbation adds calibrated noise $\mathbf{n}\in\mathbb{R}^M$ (typically multivariate Gaussian) to hide the difference vector $\bm{v} \triangleq 
 f(\bm{x})-f(\bm{x}')$
where $\bm{x}'$ is any neighboring dataset. 
Two recent studies \cite{ji-r1smg,liu2026privacy} have shown that every differentially private mechanism based on noise perturbation implicitly forms a random triangle. To be more specific,   let  $\bm{x}, \bm{x}' \in \mathbb{N}^{|\mathcal{X}|}$ be a pair of neighboring datasets,   $f(\cdot)\in\R^M$ is a query or computation, and   $\mathbf{n}$ and $\mathbf{n}'$ are the noise vectors    used to obscure $f(\bm{x})$ and $f(\bm{x}')$, respectively.   Then, for  any noise perturbation mechanism, $\mathcal{M}$, that produces a random outcome  $\mathbf{s}\in\mathcal{S}\subset\R^M$, 
we must have the following
\begin{equation}\label{eq:triangle-eq}
    \begin{aligned}
        \mathcal{M}(\bm{x}) = \mathcal{M}(\bm{x}') =\mathbf{s}\in\mathcal{S}  
        \Leftrightarrow  &f(\bm{x})+\mathbf{n}= f(\bm{x}')+\mathbf{n}'=\mathbf{s}\in\mathcal{S}\\
        \Leftrightarrow &\ \mathbf{n}' - \mathbf{n} = f(\bm{x})-f(\bm{x}') = \bm{v}.
    \end{aligned}
\end{equation}
Clearly, (\ref{eq:triangle-eq}) implies that $\bm{v} \triangleq f(\bm{x})-f(\bm{x}')$, $\mathbf{n}$ and  $\mathbf{n}'$  all belong to a specific hyperplane. 
In particular, $\bm{v}$, $\mathbf{n}$, and $\mathbf{n}'$ must   form a   \textbf{random} triangle, i.e., $\mathbf{n}' - \mathbf{n} = \bm{v}$, as shown in  Figure~\ref{fig:geo-dp}(a).

In this work, we focus the discussion on Gaussian perturbation noise.   By considering varying values of the privacy parameters, $\epsilon$ and $\delta$, and different dimension of $f(\bm{x})\in\R^M$, we   show 300 random triangles generated from 300 independent Gaussian mechanism instances   in Figure~\ref{fig:geo-dp}(b), where the blue edges denote the difference vectors $\bm{v}$, and the red edges represent $\mathbf{n}$ and $\mathbf{n}'$. Without loss of generality, we consider $l_2$ sensitivity $\Delta_2f = 1$, i.e., the worst case scenario is $\sup_{\boldsymbol{x}\sim\boldsymbol{x}'}\|\bm{v}\|_2=1$. 
To fit all 300 triangles  in one plot, we have shifted the triangles and centered them at different coordinates on the canvas. Thus, only the shape matters, not their positions.







Specifically,  Figure~\ref{fig:geo-dp}(b) illustrates how the shape of  induced random triangles evolves as privacy parameters ($\epsilon$ and $\delta$) and dimensionality ($M$) vary. Along the  horizontal axis, as privacy guarantee increases (i.e.,   $\epsilon$ and $\delta$ decrease),   the triangles become narrow and more concentrated. 
Similar pattern is also observed along the vertical axis, as dimensionality $M$ increases,   the triangles also become more elongated.  This visual pattern highlights the key effects: both higher privacy and higher dimension distort  the triangles and flatten  them,   making the geometric structure increasingly dominated by the perturbation noise. 

In what follows, we introduce two geometric views that (i) provide insight into the privacy guarantees, and (ii) characterize how the random triangles behave under different privacy parameters and ambient dimensions.

\subsection{Embedding Random Triangles into $x+y+z=1$ Plane}\label{sec:task1.1}

 The first geometric tool maps each random triangle (e.g., in Figure~\ref{fig:geo-dp}(b)) to a normalized tuple 
$(x,y,z)$, i.e., a point lying on the plane of 
$x+y+z=1$.  This planar   tool     provides an intuitive way to understand the association between privacy, utility, and dimensionality of the problem.

According to  Theorem~\ref{thm:gaussian_privacyguarantee},    $(\epsilon,\delta)$-DP is achieved if one releases the sanitized result $f(\bm{x})+\mathbf{n}\in\R^M$, where the noise vector $\mathbf{n}$ is attributed to a multivariate Gaussian distribution $\mathcal{N}(\mathbf{0},\sigma^2\bm{I})$ with $\sigma^2=\frac{2 \log(1.25/\delta)(\Delta_2f)^2}{\epsilon^2}$, and $\Delta_2f$ is the    $l_2$ sensitivity of a computation   $f(\cdot)$, i.e., $\Delta_2f \triangleq \sup_{\bm{x}\sim\bm{x}'}\|f(\bm{x})-f(\bm{x}')\|_2$. 
Clearly, $\sigma^2$ increases as $\Delta_2f$ or $\frac{1}{\epsilon}$ increase. In other words, the magnitude of the perturbation noise, $\|\mathbf{n}\|_2$, is expected to increase when dealing with high sensitivity values (large $\Delta_2f$) or strict privacy requirement (small $\epsilon$). Additionally, the expected magnitude of the additive noise, i.e., $\sigma^2M$, also increases as the dimensionality of the problem, $M$. 

Hence, we propose to study  the association between $\Delta_2f$, $\epsilon$, and $M$ by investigating the \textbf{normalized random triangle} $\widetilde{\triangle}$ defined as   
\begin{equation}\label{eq:normalized_triangle}
    \widetilde{\triangle} \triangleq [\bm{e}_1\ \  \bm{e}_2\ \  \bm{e}_3 ]=\frac{[\bm{v}\ \  \mathbf{n} \ \  (-\bm{v}-\mathbf{n})]}{ \|\bm{v}\ \  \mathbf{n} \ \  (-\bm{v}-\mathbf{n})\|_F},
\end{equation}
where $\|\cdot\|_F$ stands for the Frobenius norm. (\ref{eq:normalized_triangle}) represents the edge vectors, i.e., $\bm{e}_1$, $\bm{e}_2$, and $\bm{e}_3$, of the normalized triangle, whose columns adds to zeros column $\mathbf{0}\in\R^M$. Due to the normalization, we have $\|\bm{e}_1\|_2^2+\|\bm{e}_2\|_2^2+\|\bm{e}_3\|_2^2=1$, and we can treat the tuple 
\begin{equation}
    \label{eq:tuple}
    \left(\|\bm{e}_1\|_2^2, \|\bm{e}_2\|_2^2, \|\bm{e}_3\|_2^2\right)\triangleq (x,y,z)
\end{equation}
as a random point on   $x+y+z=1$ plane, where $x>0, y>0, z>0$.

\begin{table*}[!htbp]
    \centering
    \begin{tabular}{c|c|c|c}
     & $M=3$ & $M=9$ & $M=15$ \\ \hline

 $\epsilon=2$ &
        \adjustbox{trim=20pt 10pt 18pt 10pt,clip}{\includegraphics[width=0.35\textwidth]{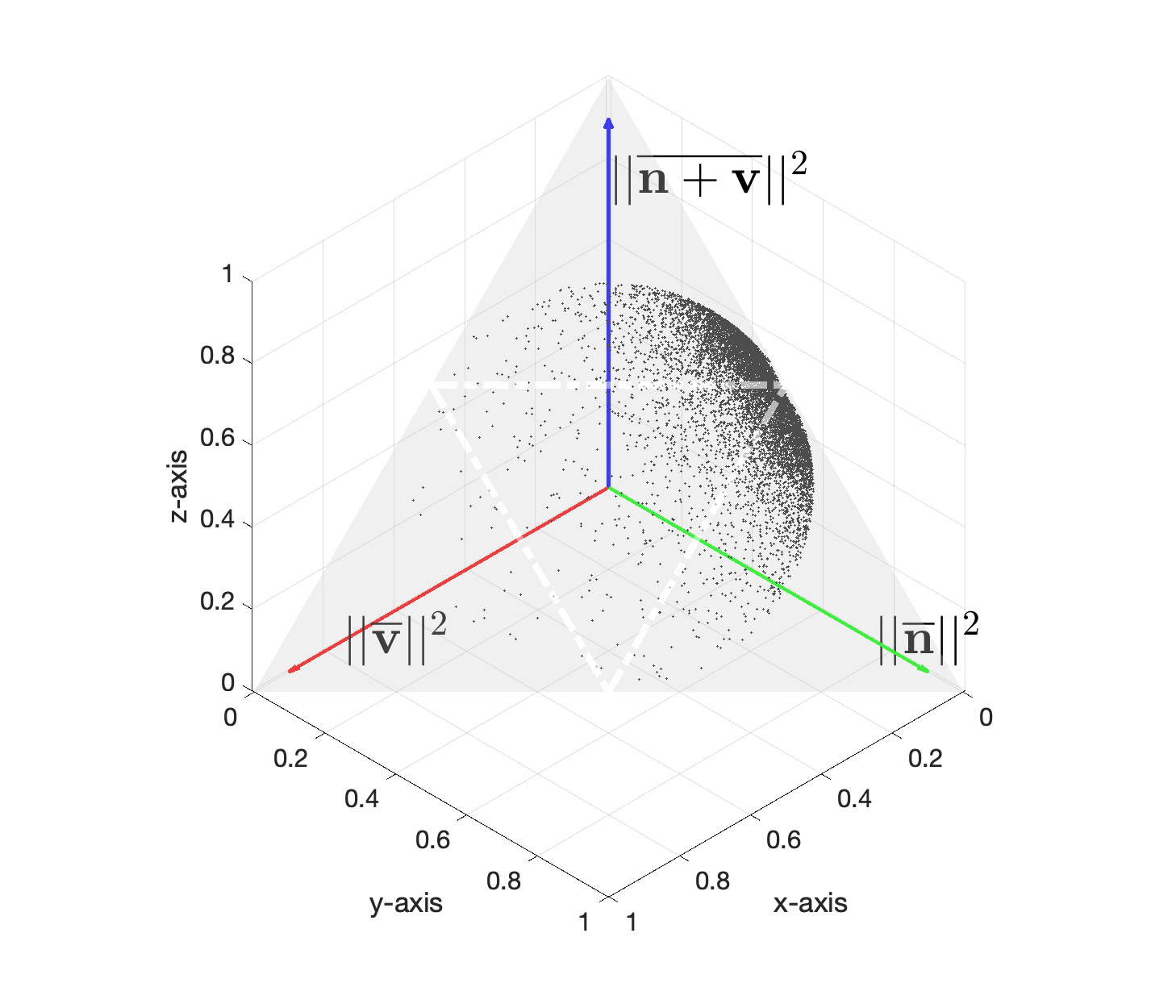}} &
        \adjustbox{trim=20pt 10pt 18pt 10pt,clip}{\includegraphics[width=0.35\textwidth]{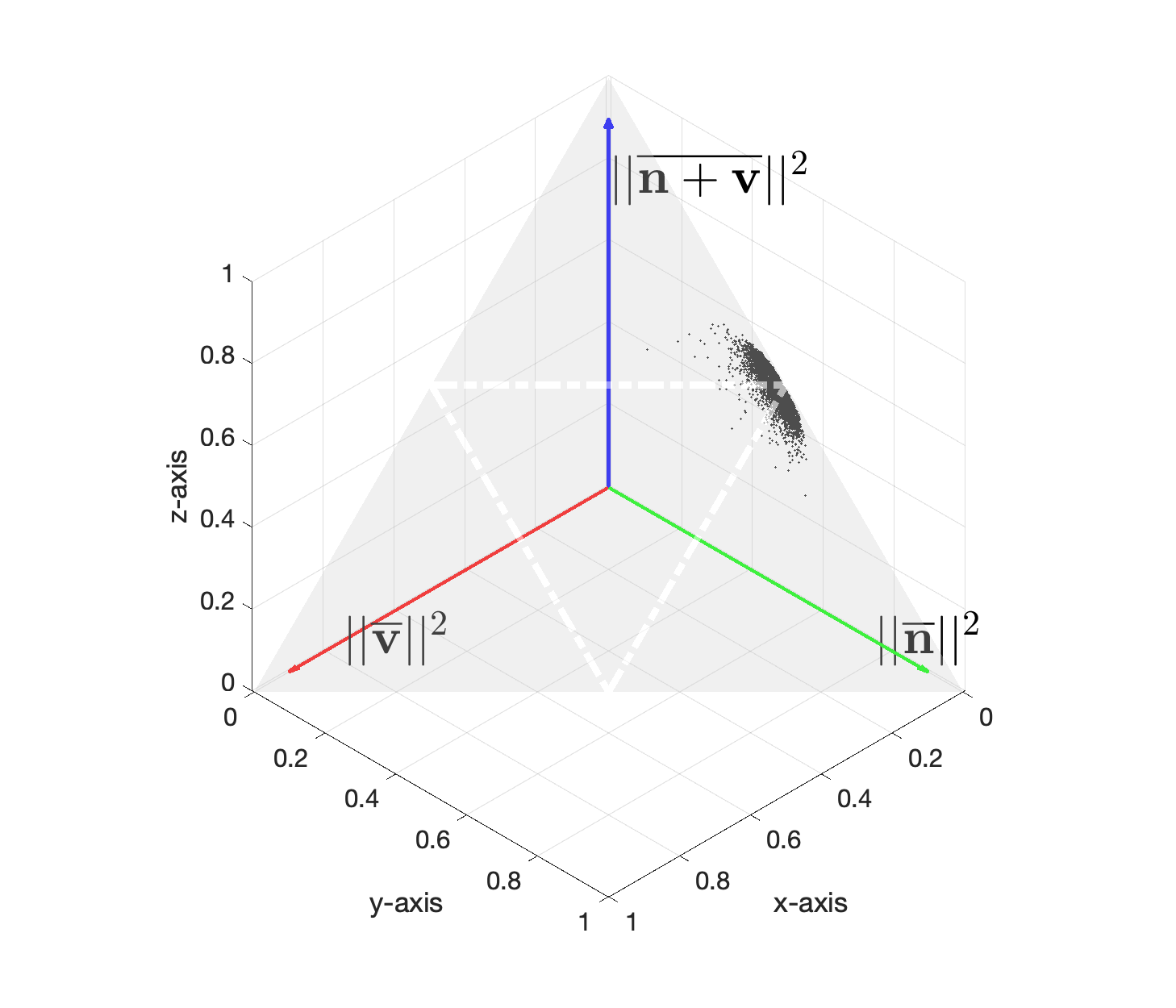}} &
        \adjustbox{trim=20pt 10pt 18pt 10pt,clip}{\includegraphics[width=0.35\textwidth]{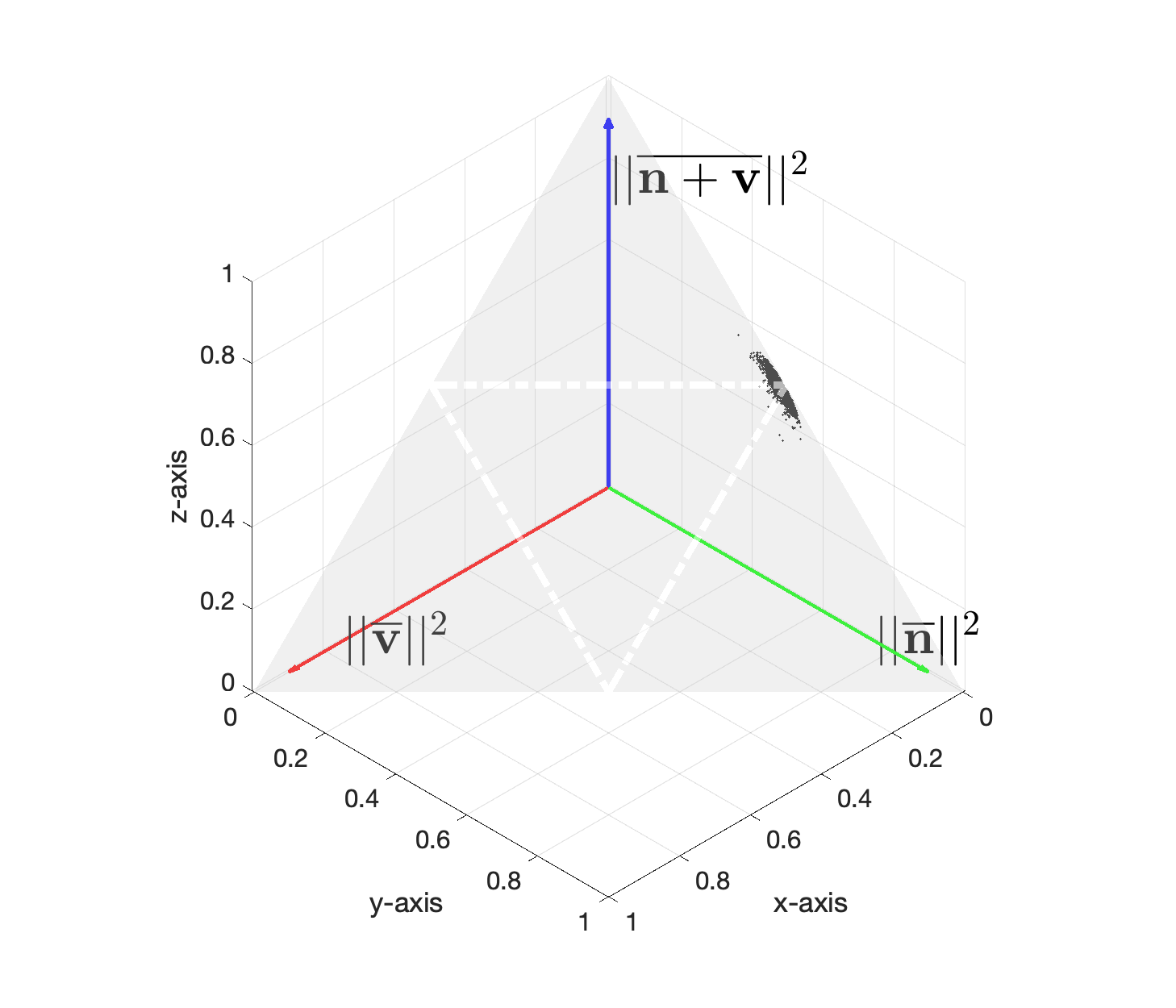}} \\\hline

    $\epsilon=6$ &
        \adjustbox{trim=20pt 10pt 18pt 10pt,clip}{\includegraphics[width=0.35\textwidth]{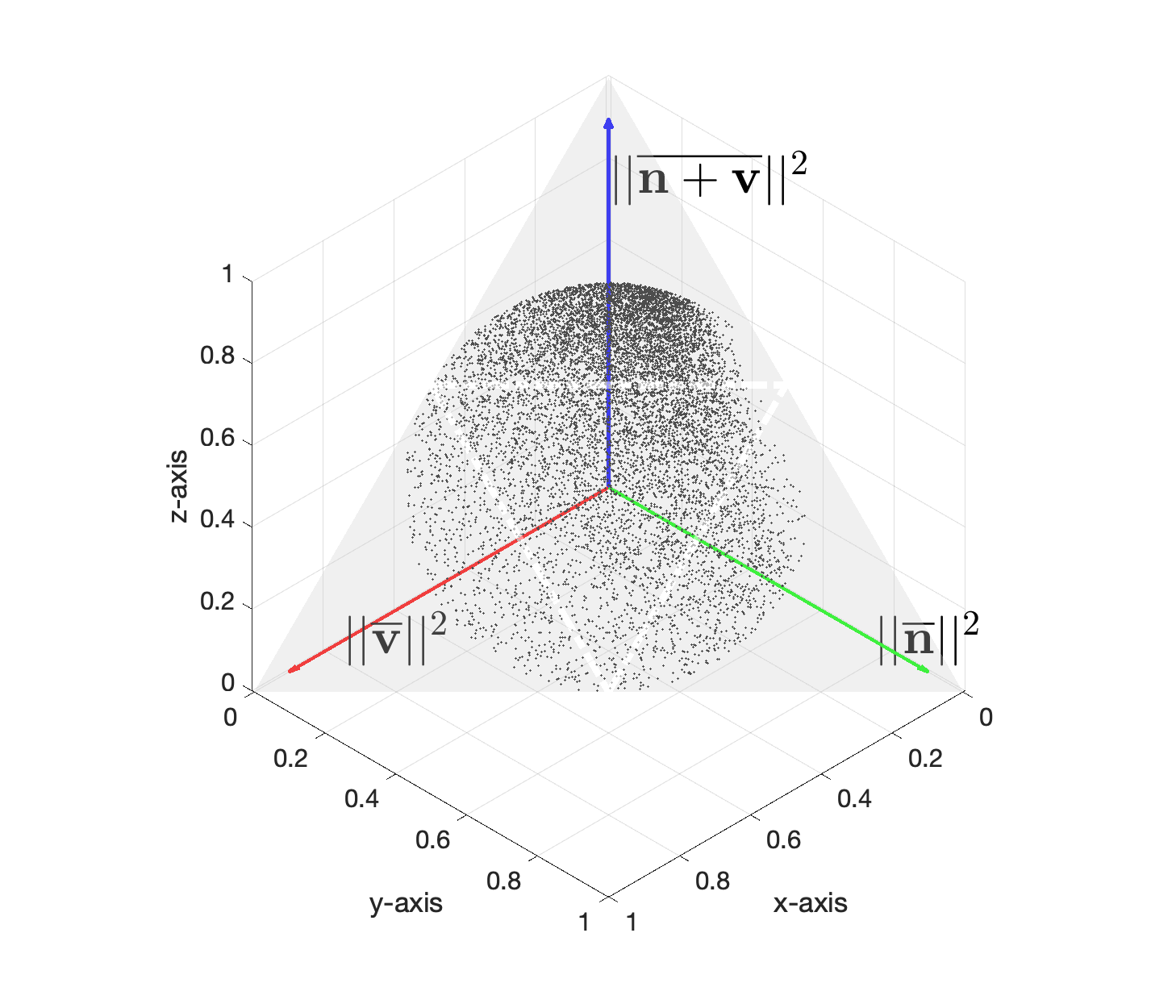}} &
        \adjustbox{trim=20pt 10pt 18pt 10pt,clip}{\includegraphics[width=0.35\textwidth]{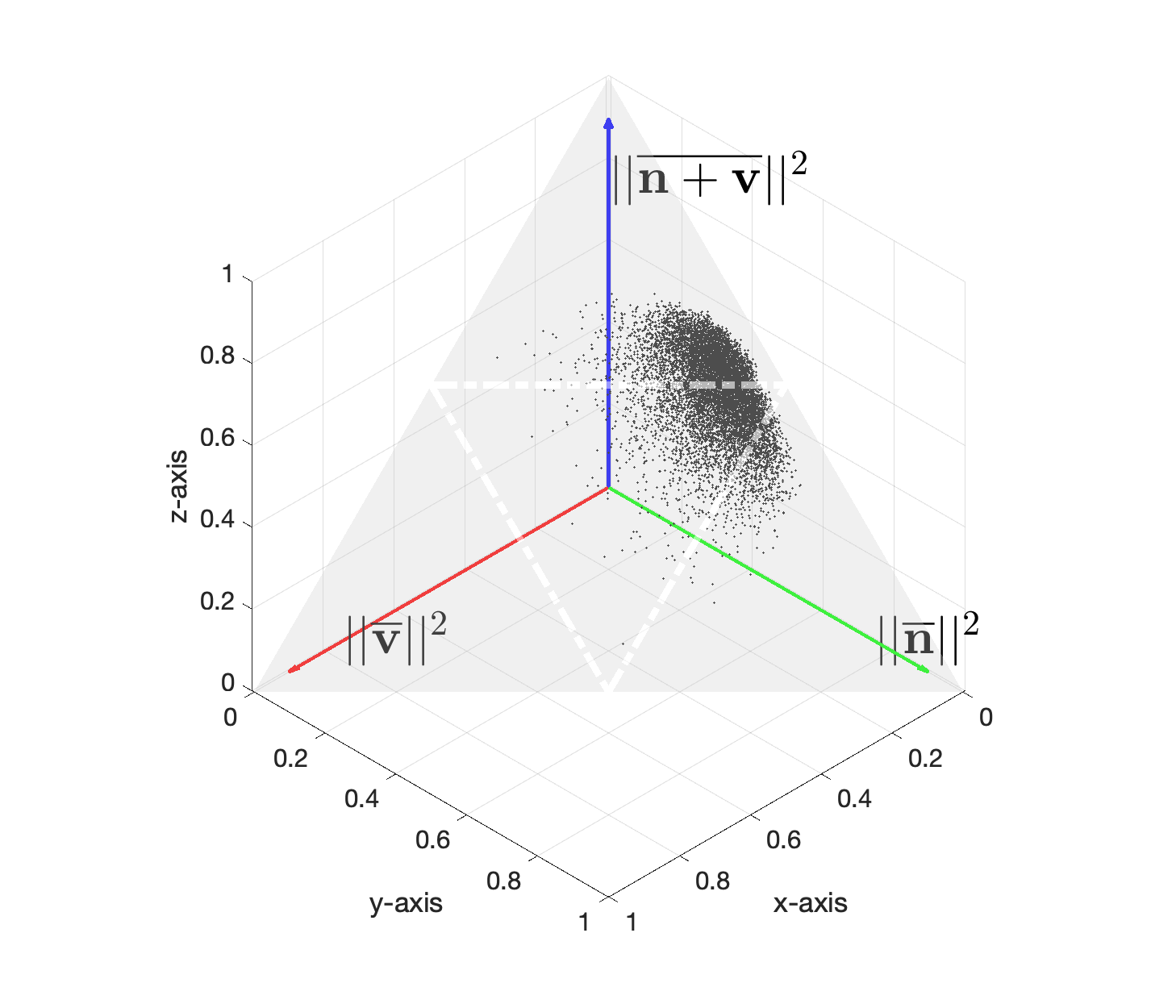}} &
        \adjustbox{trim=20pt 10pt 18pt 10pt,clip}{\includegraphics[width=0.35\textwidth]{images/xyz-plane-view/plane-eps2-d15.eps}} \\\hline

 $\epsilon=10$ &
        \adjustbox{trim=20pt 10pt 18pt 10pt,clip}{\includegraphics[width=0.35\textwidth]{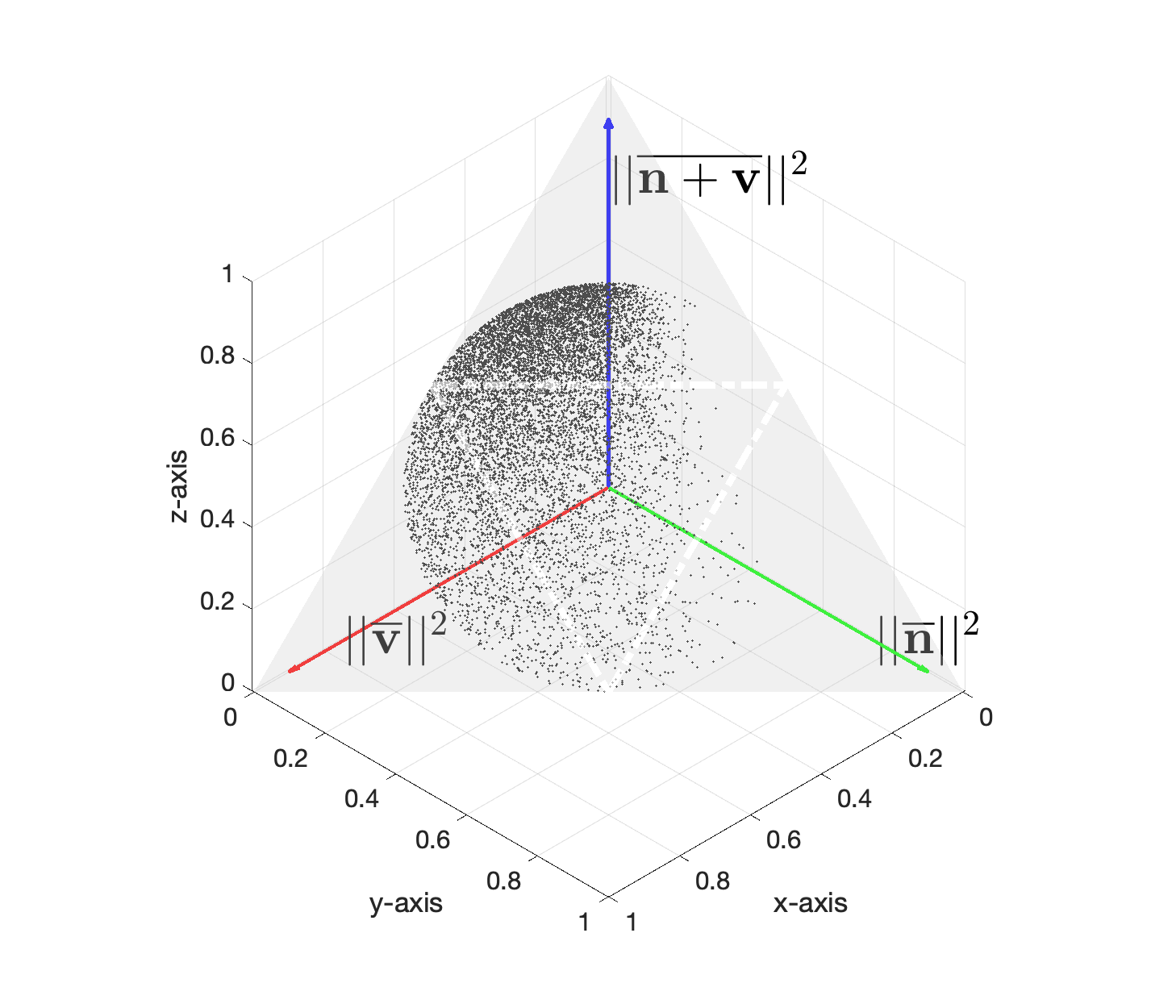}} &
        \adjustbox{trim=20pt 10pt 18pt 10pt,clip}{\includegraphics[width=0.35\textwidth]{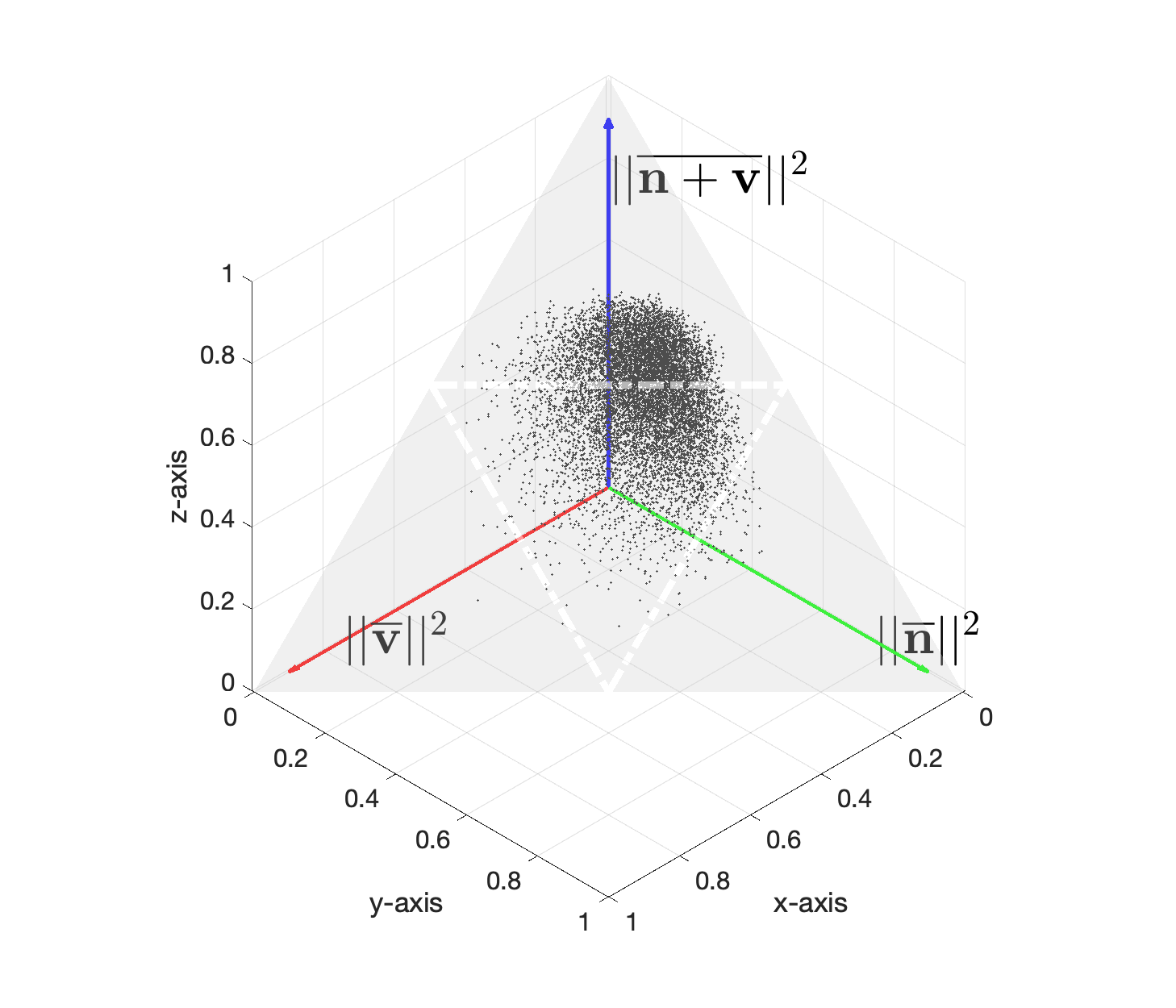}} &
        \adjustbox{trim=20pt 10pt 18pt 10pt,clip}{\includegraphics[width=0.35\textwidth]{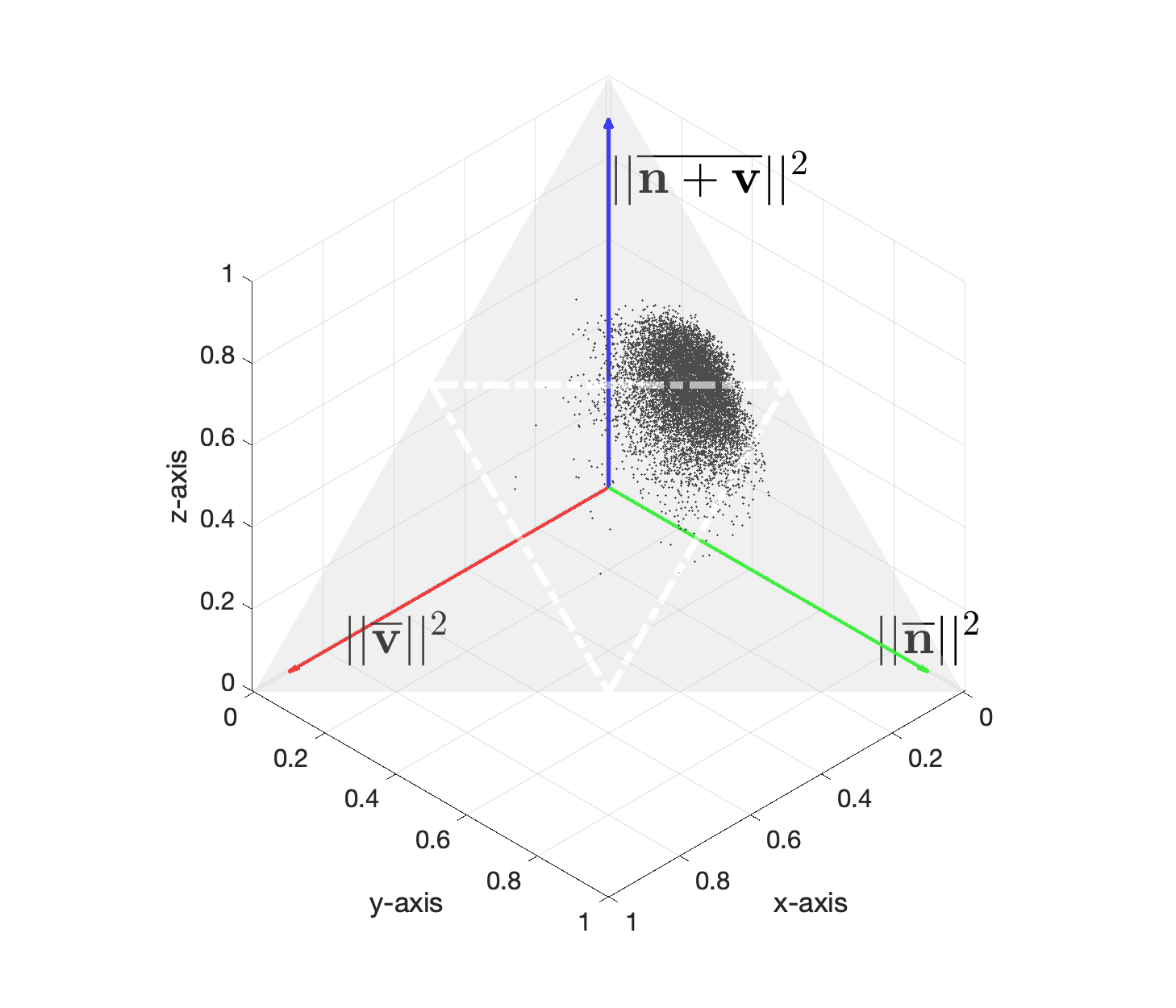}} \\
    \end{tabular}
\caption{10,000 point clouds of $\left(\|\bm{e}_1\|_2^2, \|\bm{e}_2\|_2^2, \|\bm{e}_3\|_2^2\right)$  on   plane          $x+y+z=1$, when    $M\in\{3,9,15\}$, $\epsilon\in\{2,10\}$,   $\delta = 10^{-5}$, and $\Delta_2f=1$}
    \label{tab:xyz-plane}
\end{table*}

In Table~\ref{tab:xyz-plane}, by considering        dimension  $M\in\{3,9,15\}$, $\epsilon\in\{2,10\}$,   $\delta = 10^{-5}$, and $\Delta_2f=1$, we scatter  the point clouds generated by 10,000 random tuples  $\left(\|\bm{e}_1\|_2^2, \|\bm{e}_2\|_2^2, \|\bm{e}_3\|_2^2\right)$  for a given combination of $M$ and $\epsilon$. Each tuple is represented by a point on the $x+y+z=1$ plane, in which the normalized edge length squared  $\|\overline{\bm{v}}\|_2^2$, $\|\overline{\mathbf{n}}\|_2^2$, and $\|\overline{\bm{v}+\mathbf{n}}\|_2^2$ determine the coordinates on the  $x$-, $y$-, and $z$-axes, respectively. One each of the $x+y+z=1$ plane, the white dotted lines represent regimes of   right triangles, for instance, $\|\bm{e}_1\|_2^2+\|\bm{e}_2\|_2^2=\frac{1}{2}=\|\bm{e}_3\|_2^2$. Note that normalization in (\ref{eq:normalized_triangle}) does not discard any information relevant to the utility–privacy trade-off.  Prior work shows that utility loss depends exclusively on the noise magnitude $\|\mathbf{n}\|_2$ \cite{ji-r1smg,hardt2010geometry,nikolov2013geometry}.  Because the normalization rescales  each    edge in one triangle  by the same Frobenius norm, it leaves the key ratio between $\|\mathbf{n}\|_2$ and $\Delta_2f$ unchanged, and therefore preserves every quantity that matters for both privacy analysis and utility accounting.

The following theorem (proved in Appendix~\ref{app:simplex-proof}) characterizes the behavior of the point clouds on the simplex. Because of the normalization in (\ref{eq:normalized_triangle}), we always have $x+y+z=1$, and the triple $(x,y,z)$ is esentially constrained to reside on a 2-dimensional simplex, which means that the distribution of $(x,y,z)$ is not a full 3D density, but a density on a 2D surface. Hence, we can characterize the randomness of the triple  with a joint PDF of two coordinates, $(x,y)$, and the third coordinate is determined via $z = 1-x-y$.

\begin{theorem}\label{thm:simplex-density} Let $M$ be the dimension of the sensitive query $f(\bm{x})\in\R^M$, $\Delta_2f$ is the $l_2$ sensitivity of the query, and i.i.d. Gaussian noise with variance $\sigma^2$ is used to obscure $f(\bm{x})$. Define the normalized triangle edges $(x,y,z)$   as in \eqref{eq:normalized_triangle}
and 
\eqref{eq:tuple}, and let  $p(x,y)$ and $q(x,y)$ be defined as 
    \begin{equation}\label{eq:p-q}
\left\{
\begin{aligned}
a \triangleq p(x,y) = & \frac{\Delta_2f (1-2x-2y)}{2x}\\
r^2 \triangleq q(x,y) = & \frac{\left(\Delta_2f\right)^2 y}{x}-\big(p(x,y)\big)^2
\end{aligned},
\right.
\end{equation}
then the joint distribution of $(x,y)$ is
\begin{equation}\label{eq:jointpdf}
\begin{aligned}
f_{X,Y}(x,y)
&=\frac{(\Delta_2 f)^3}{2x^3}
\cdot
\frac{1}{\sqrt{2\pi}\,\sigma}
\exp\!\left[-\frac{1}{2\sigma^2}p(x,y)^2\right]
\cdot 
\frac{\big(q(x,y)\big)^{\frac{M-3}{2}}}
{2^{\frac{M-1}{2}}
\sigma^{M-1}
\Gamma\!\left(\frac{M-1}{2}\right)}
\exp\!\left[-\frac{1}{2\sigma^2}q(x,y)\right] \\
\end{aligned},
\end{equation}
where $0<x<1$, $0<y<1$, $x+y<1$, $q(x,y)>0$, and $z$ is uniquely obtained via $z=1-x-y$.
\end{theorem}



Subsequently, we have the following corollary (proved in Appendix~\ref{proof-simplex-plrv}) that connects the traditional PLRV in \eqref{eq:plrv_reduction} to the PLRV supported on the simplex. 

\begin{corollary}\label{corollary:new-plrv} The PLRV associated with the classic Gaussian mechanism in (\ref{eq:plrv_reduction}) can be reconstructed by a random tuple  $(x,y)$ generated from 
      (\ref{eq:jointpdf}) via 
  \begin{equation}\label{eq:plrv-simplex}
    \mathrm{PLRV}   \triangleq   L (x,y) = \frac{\left(\Delta_2f\right)^2}{2\sigma^2} \cdot \frac{x+2y-1}{x}.
  \end{equation}    
\end{corollary}

\begin{remark}  Note that \eqref{eq:plrv-simplex} 
is a reparameterization of the classic PLRV in \eqref{eq:v-plrv} using coordinates on the simplex. Thus, the two PLRVs should exhibit  the same tail probability bound. In other words, performing measure concentration analysis on $f_{X,Y}(x,y)$ will lead to the same DP guarantee with the Gaussian mechanism. While the analysis of classic Gaussian mechanism relies on one-dimensional projections, i.e., $\langle \mathbf{n}, \bm{v} \rangle$, our simplex mapping reveals the intrinsic high-dimensional geometry of the privacy mechanism. 
\end{remark}


The simplex representation reveals two distinct geometric properties of the
normalized random triangles induced by the Gaussian mechanism. First, the
triangle-feasibility condition determines an exact elliptical support for the
simplex coordinates. This support is independent of the ambient dimension
$M$. Second, as $M$ increases, the probability mass within this admissible
region becomes increasingly concentrated near a degenerate,
noise-dominated configuration. In particular, we have the following proposition which is proved in  
Appendix~\ref{app:proof-phase_transition}.

\begin{prop}[Elliptical support and high-dimensional collapse]\label{prop:phase_transition}
Consider the simplex representation of the Gaussian mechanism defined \eqref{eq:normalized_triangle} and \eqref{eq:tuple}, with
joint density $f_{X,Y}(x,y)$.
\begin{enumerate}
    \item \textbf{Elliptical support.}
    The support of $(X,Y)$ is contained in
    \begin{equation*}
        \mathcal{S}_{\mathrm{ell}}
      \triangleq
        \left\{
        (x,y):
        \begin{array}{l}
            x>0,\quad y>0,\quad x+y<1,\\[1mm]
            4x^2+4xy+4y^2-4x-4y+1<0
        \end{array}
        \right\},
        \label{eq:elliptical-support}
    \end{equation*}
which is a rotated ellipse centered at $(1/3,1/3)$  in the $(x,y)$-coordinates.

    \item \textbf{High-dimensional collapse.}
    As $M\to\infty$, the   simplex coordinates satisfy $x=\Theta\!\left(\frac1M\right)$, $y=\frac12+O\!\left(\frac1M\right)$, and $z=\frac12+O\!\left(\frac1M\right)$. 
    Consequently, the normalized triangle shape concentrates toward $(x,y,z)
        \longrightarrow
        \left(0,\frac12,\frac12\right)$.
\end{enumerate}
\end{prop}
\noindent\textbf{Geometric interpretation.}
The simplex representation separates an exact finite-dimensional support
property from a dimension-dependent concentration phenomenon. The
elliptical support describes all feasible normalized random triangle 
configurations, while the high-dimensional limit shows that Gaussian-induced
triangles become increasingly dominated by their two noise-related edges.
Although the scalar PLRV can be reconstructed from $(x,y)$, the full simplex
distribution additionally describes how normalized triangle shapes are
distributed over its level sets, i.e., $\{(x,y):L(x,y)=l\}$, where $L(x,y)$ is defined in~\eqref{eq:plrv-simplex} and $l>0$ is a certain privacy loss value.

\subsection{Embedding Random Triangles into   Hemisphere}

The simplex representation describes a normalized triangle through its
three squared edge lengths. Now, we   introduce a complementary representation based on
the singular values and principal directions of the normalized triangle 
via the following two steps.

\textbf{Step 1: A unique SVD of $\widetilde{\triangle}$.} Given the  normalized random triangle $\widetilde{\triangle}$ defined in \eqref{eq:normalized_triangle}, we let $\widetilde{\triangle} = \bm{U}\bm{\Sigma}\bm{V}^T$ be its SVD. 
Note that the orthogonal factor $\bm{U}$ in the SVD is not important because $\widetilde{\triangle}$ and $\bm{U}^{-1}\widetilde{\triangle}$ correspond to the same normalized triangle being just rotated. 
Hence, \(\bm{U}\) can be 
quotiented out in the shape representation~\cite{kendall1989survey,small2012statistical}, and we set  $\bm{U}=\bm{I}$ (the identity matrix) and consider  $\widetilde{\triangle} = \bm{I}\bm{\Sigma}\bm{V}^T$.   Moreover, 
$\mathrm{rank}(\widetilde{\triangle})\leq 2$ since $\bm{e}_1+\bm{e}_2+\bm{e}_3=\bm{0}$.

Thus, to perform   SVD on $\widetilde{\triangle}$,  we first use the Helmert transformation $\mathbf{H}=\begin{bmatrix}
        \frac{1}{\sqrt{2}} & \frac{1}{\sqrt{6}}\\
        -\frac{1}{\sqrt{2}} & \frac{1}{\sqrt{6}}\\
        0 & -\frac{2}{\sqrt{6}}
    \end{bmatrix}$ to  map
the three centered edge coordinates to their two-dimensional column subspace.
In particular, the Helmert matrix satisfies $\mathbf{H}^\top   \mathbf{H}=  \mathbf{I}_2$, $\mathbf{H}^\top\mathbf{1}_3=0$, and $\mathbf{H}  \mathbf{H}^\top
    =
      \mathbf{I}_3-\frac{1}{3}\mathbf{1}_3\mathbf{1}_3^\top$.

The SVD of $\widetilde{\Delta}$ can therefore be applied to the reduced two-column matrix
$\widetilde{\Delta}\mathbf{H}\in\mathbb{R}^{M\times 2}$. This reduction preserves all nonzero
singular values because
\begin{equation*}
    (\widetilde{\Delta}\mathbf{H})
    (\widetilde{\Delta}\mathbf{H})^\top
    =
    \widetilde{\Delta}\mathbf{H}\mathbf{H}^\top
    \widetilde{\Delta}^\top
    = \widetilde{\Delta}\left( \mathbf{I}_3-\frac{1}{3}\mathbf{1}_3\mathbf{1}_3^\top\right)\widetilde{\Delta}^\top = 
    \widetilde{\Delta}\widetilde{\Delta}^\top,
\end{equation*}
where the last equality follows from
$\widetilde{\Delta}\mathbf{1}_3=\mathbf{0}$ (the columns adds to zeros column).
Thus, right multiplying $\widetilde{\Delta}$  by $\mathbf{H}$ removes only the redundant null
space of $\widetilde{\Delta}$ spanned by $\mathbf{1}_3$ and preserves the complete nonzero
spectral information of $\widetilde{\Delta}$. Then, according to~\cite{edelman2015random}, 
let the singular values $\sigma_1\geq\sigma_2\geq0$ and $\alpha\in[0,\pi)$, the unique SVD of $\widetilde{\triangle}\mathbf{H}$ is 
\begin{equation}\label{eq:unique-svd}
        \widetilde{\triangle}\mathbf{H} = \bm{I}\bm{\Sigma}\bm{V}^T = \bm{I} \left[\begin{matrix}
\sigma_1 & 0 \\
0 & \sigma_2  \\
\vdots & \vdots \\
0 & 0 
\end{matrix}\right]\left[\begin{matrix}
\cos \alpha & -\sin \alpha  \\
\sin \alpha & \cos \alpha  
\end{matrix}\right], \quad \bm{I}\in\R^{M\times M}, \bm{\Sigma}\in\R^{M\times 2}, \bm{V}\in\R^{2\times 2}.
\end{equation}

\textbf{Step 2: Map $\widetilde{\triangle}$ to hemisphere.} By resorting to stochastic shape theory~\cite{kendall1989survey,small2012statistical}, we   associate the non-zero entries in $\bm{\Sigma}$ ($\sigma_1$ and $\sigma_2$) and $\alpha$ with a point on the hemisphere of radius $\frac{1}{2}$, on which 
\begin{equation}\label{eq:long-lat}
    \begin{aligned}
        \text{the longitude is represented by } &  \theta = 2\alpha\in [0, 2\pi),\\
        \text{the height is represented by } & h=\sigma_1  \sigma_2\in\left[0,\frac{1}{2}\right],\\
        \text{the latitude is represented by } &  \phi = \arcsin(2\sigma_1  \sigma_2)\in\left[0,\frac{\pi}{2}\right].
    \end{aligned}
\end{equation}

\begin{table*}[!htbp]
    \centering
    \begin{tabular}{c|c|c|c}
     & $M=3$ & $M=9$ & $M=15$ \\ \hline

 $\epsilon=2$ &
        \adjustbox{trim=20pt 10pt 18pt 14pt,clip}
        {\includegraphics[width=0.35\textwidth]{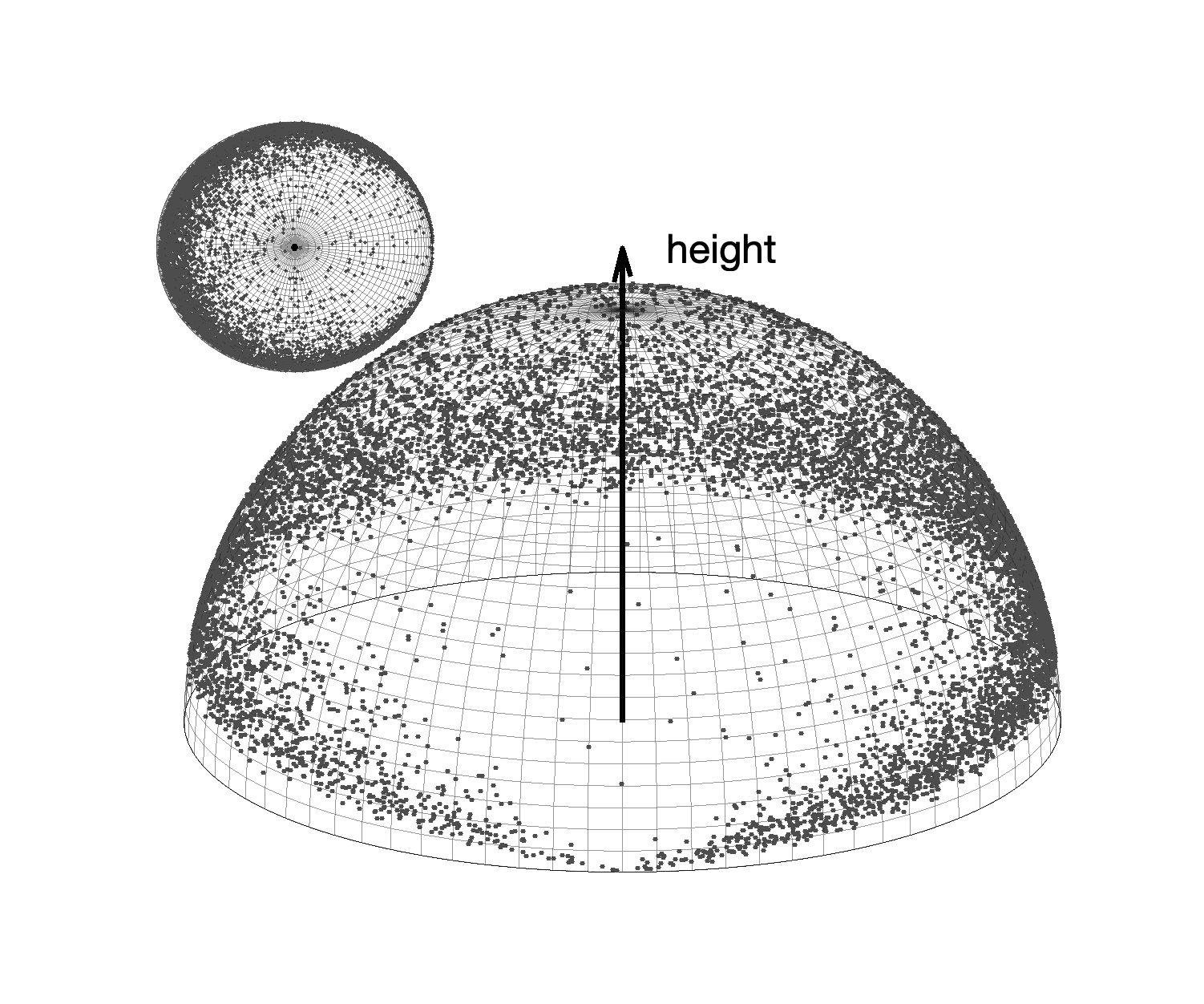}} &
        \adjustbox{trim=20pt 10pt 18pt 14pt,clip}{\includegraphics[width=0.35\textwidth]{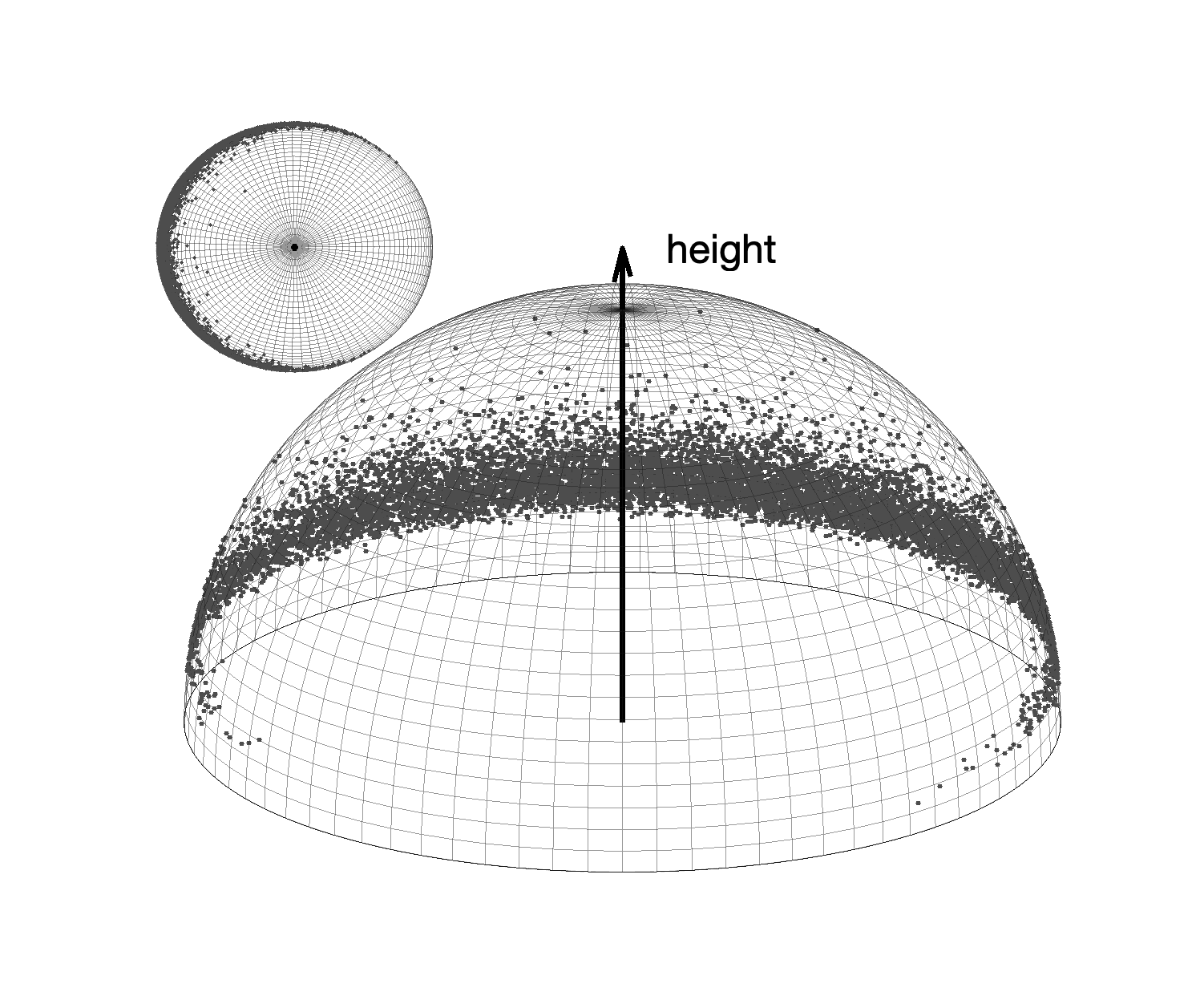}} &
        \adjustbox{trim=20pt 10pt 18pt 14pt,clip}{\includegraphics[width=0.35\textwidth]{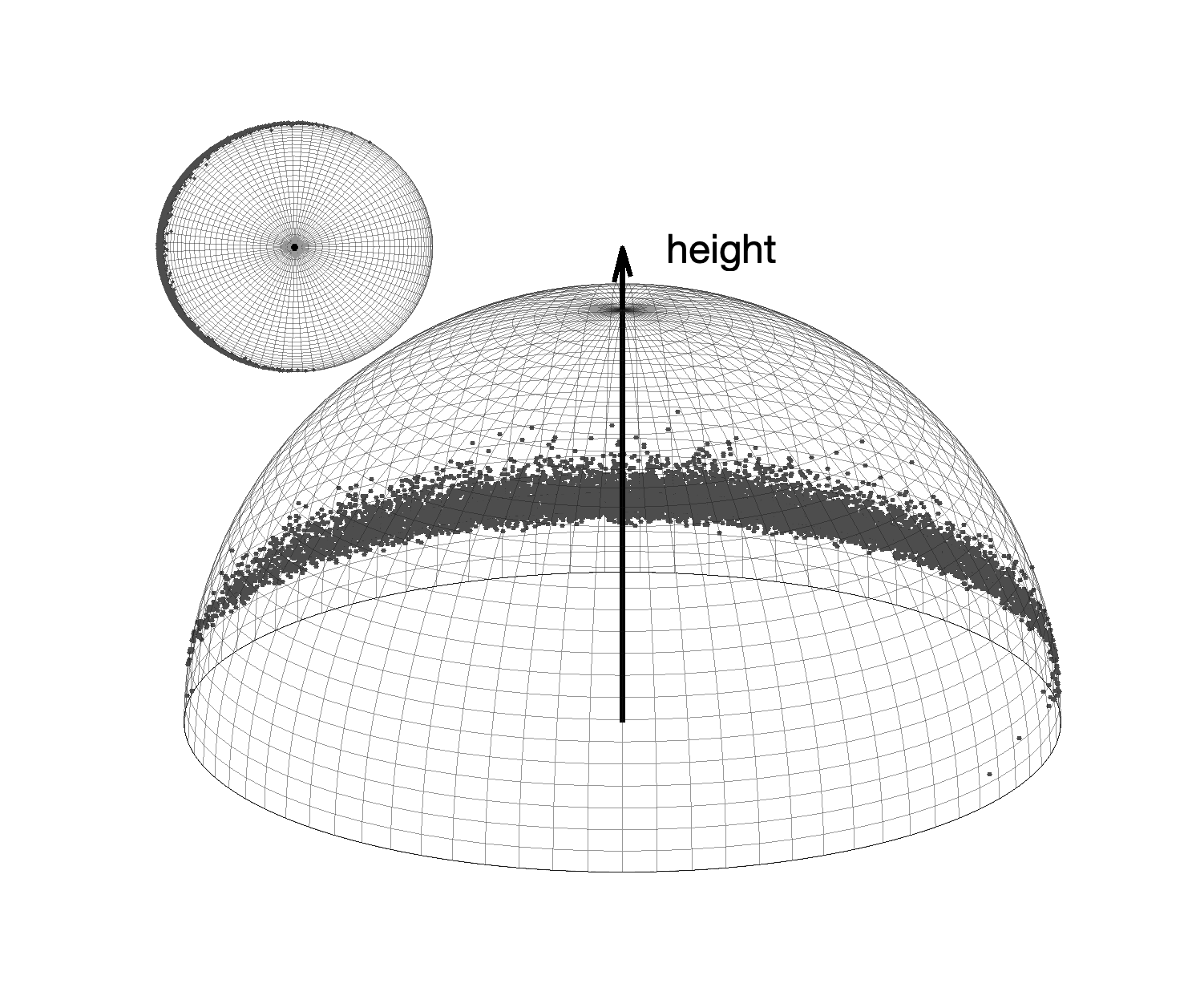}} \\\hline

    $\epsilon=6$ &
        \adjustbox{trim=20pt 10pt 18pt 14pt,clip}{\includegraphics[width=0.35\textwidth]{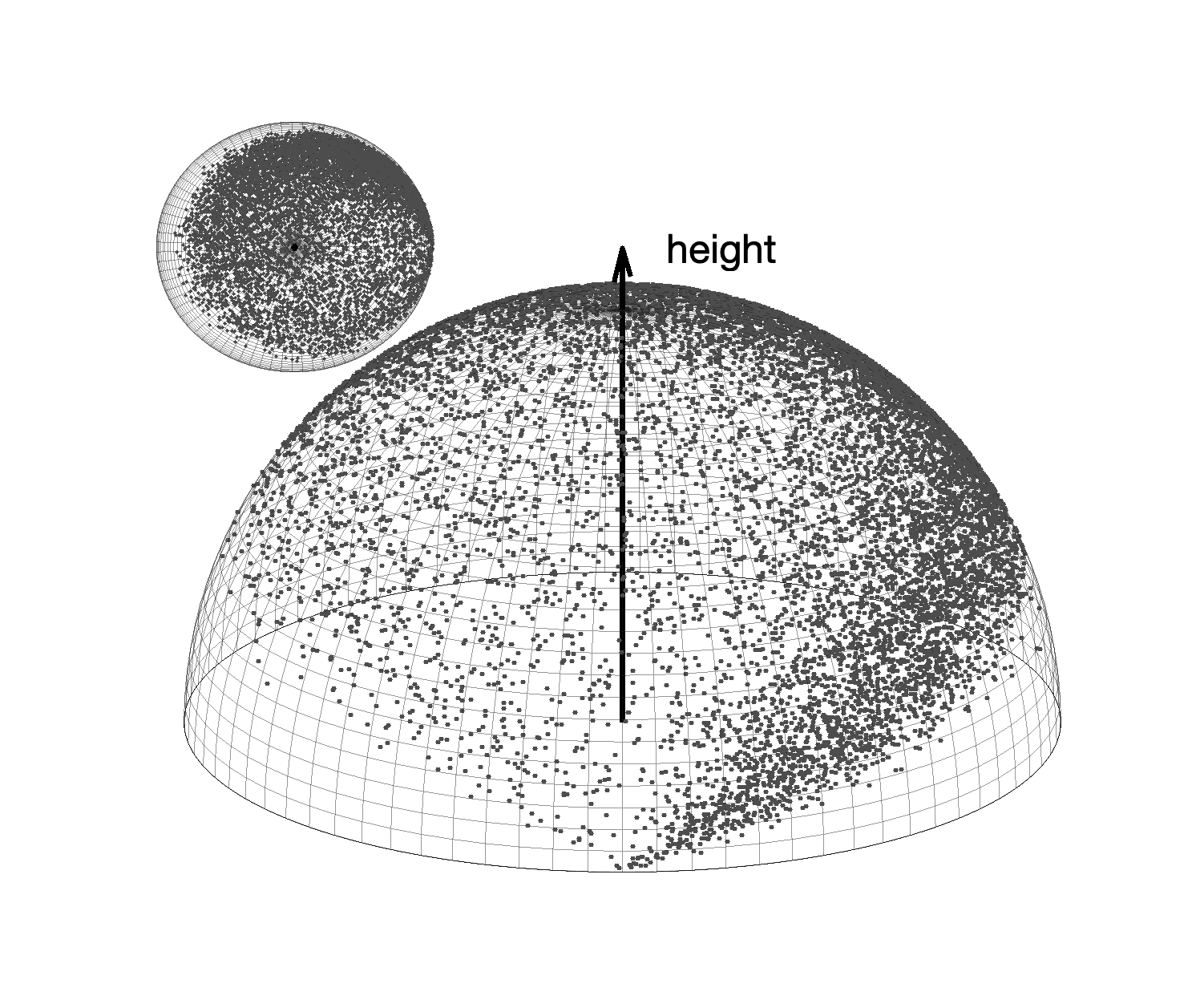}} &
        \adjustbox{trim=20pt 10pt 18pt 14pt,clip}{\includegraphics[width=0.35\textwidth]{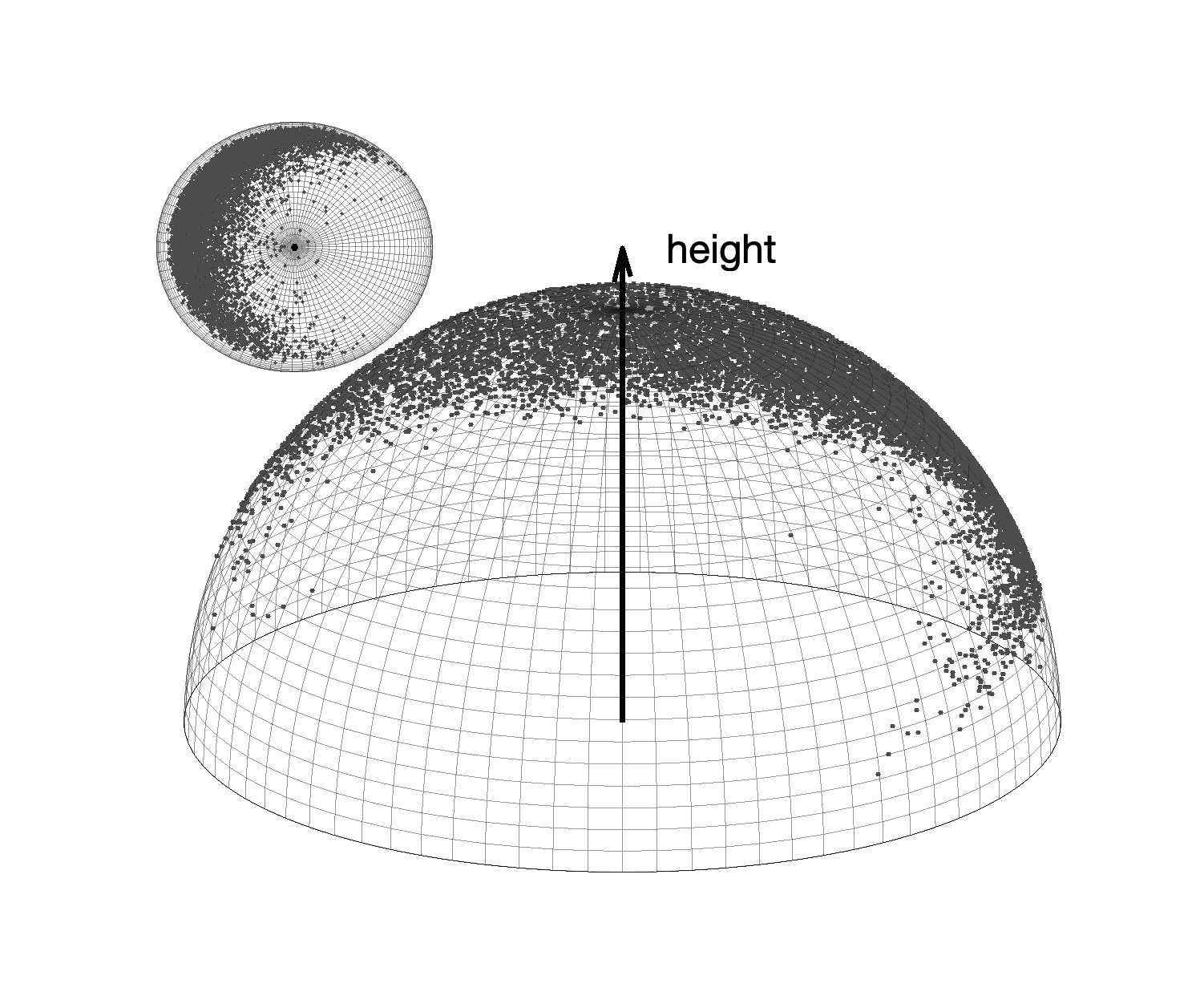}} &
        \adjustbox{trim=20pt 10pt 18pt 14pt,clip}{\includegraphics[width=0.35\textwidth]{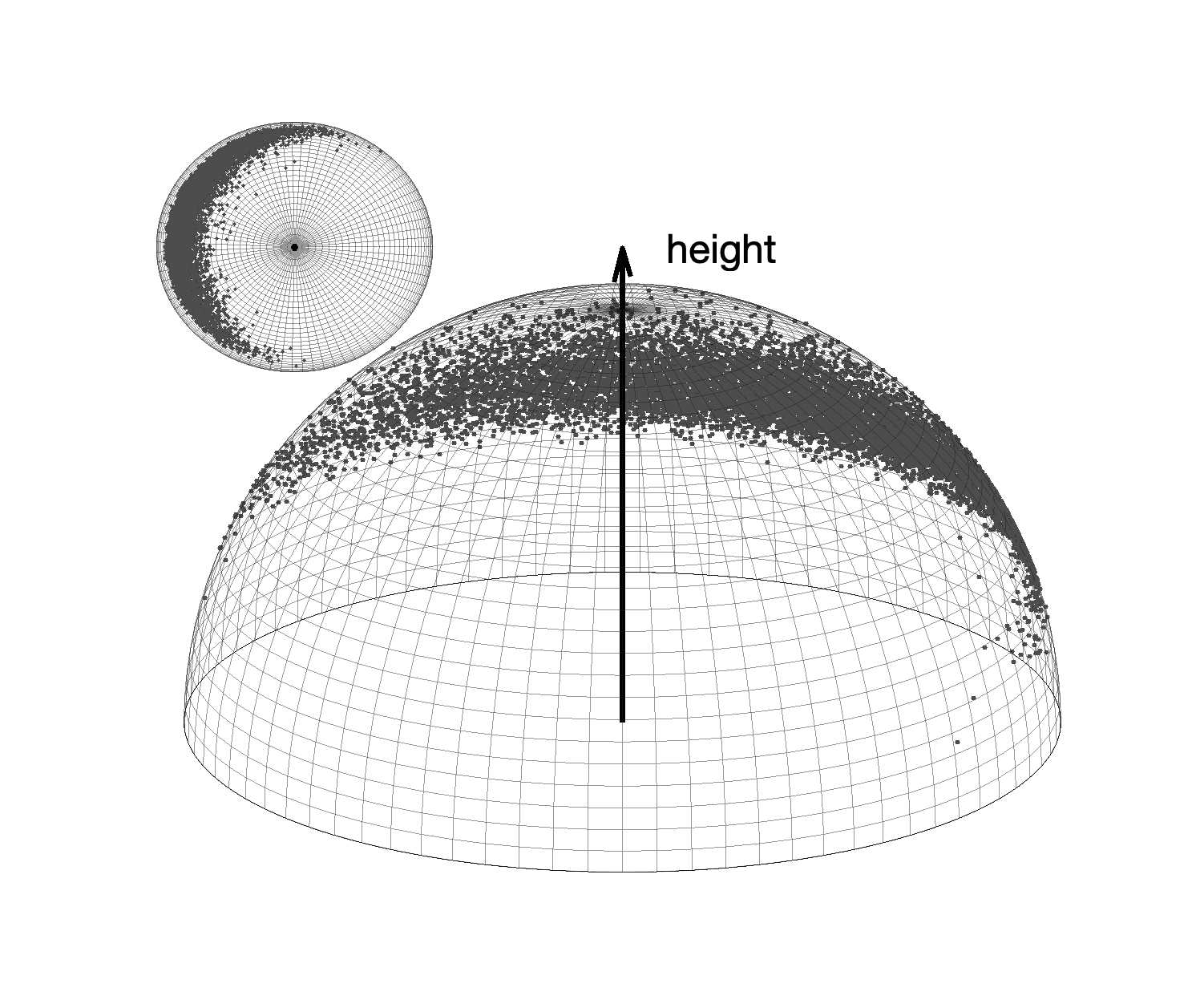}} \\\hline

 $\epsilon=10$ &
        \adjustbox{trim=20pt 10pt 18pt 14pt,clip}{\includegraphics[width=0.35\textwidth]{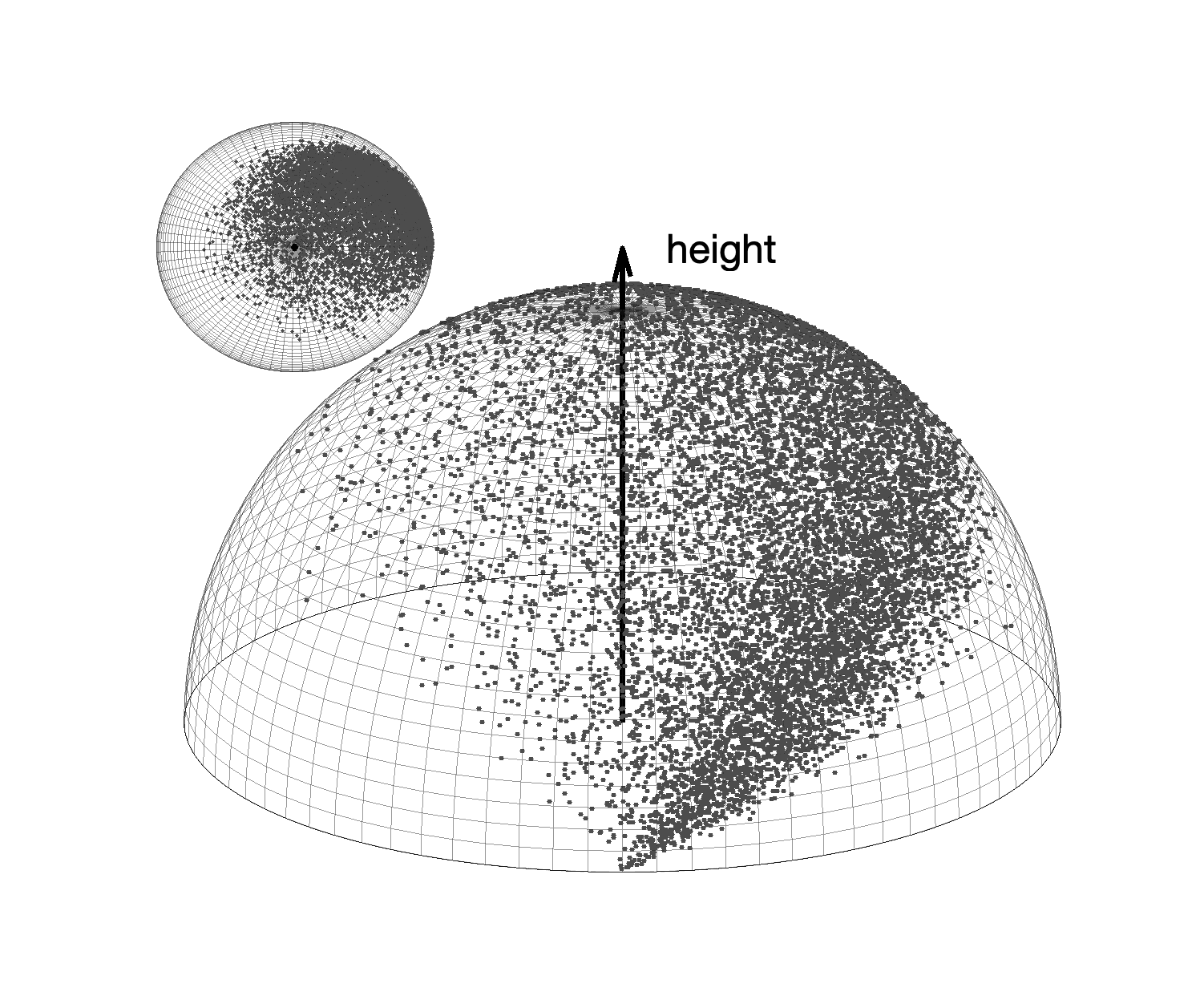}} &
        \adjustbox{trim=20pt 10pt 18pt 14pt,clip}{\includegraphics[width=0.35\textwidth]{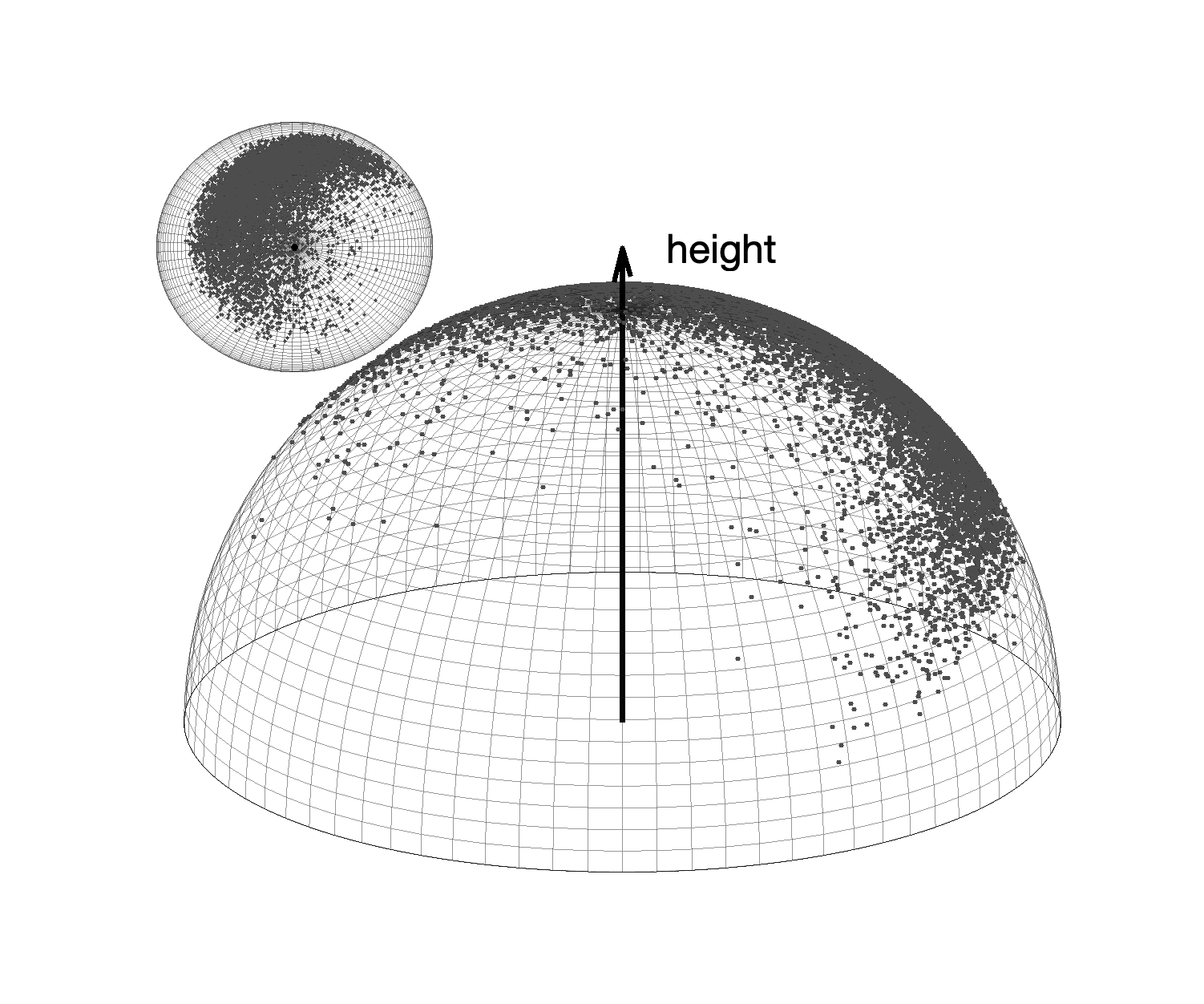}} &
        \adjustbox{trim=20pt 10pt 18pt 14pt,clip}{\includegraphics[width=0.35\textwidth]{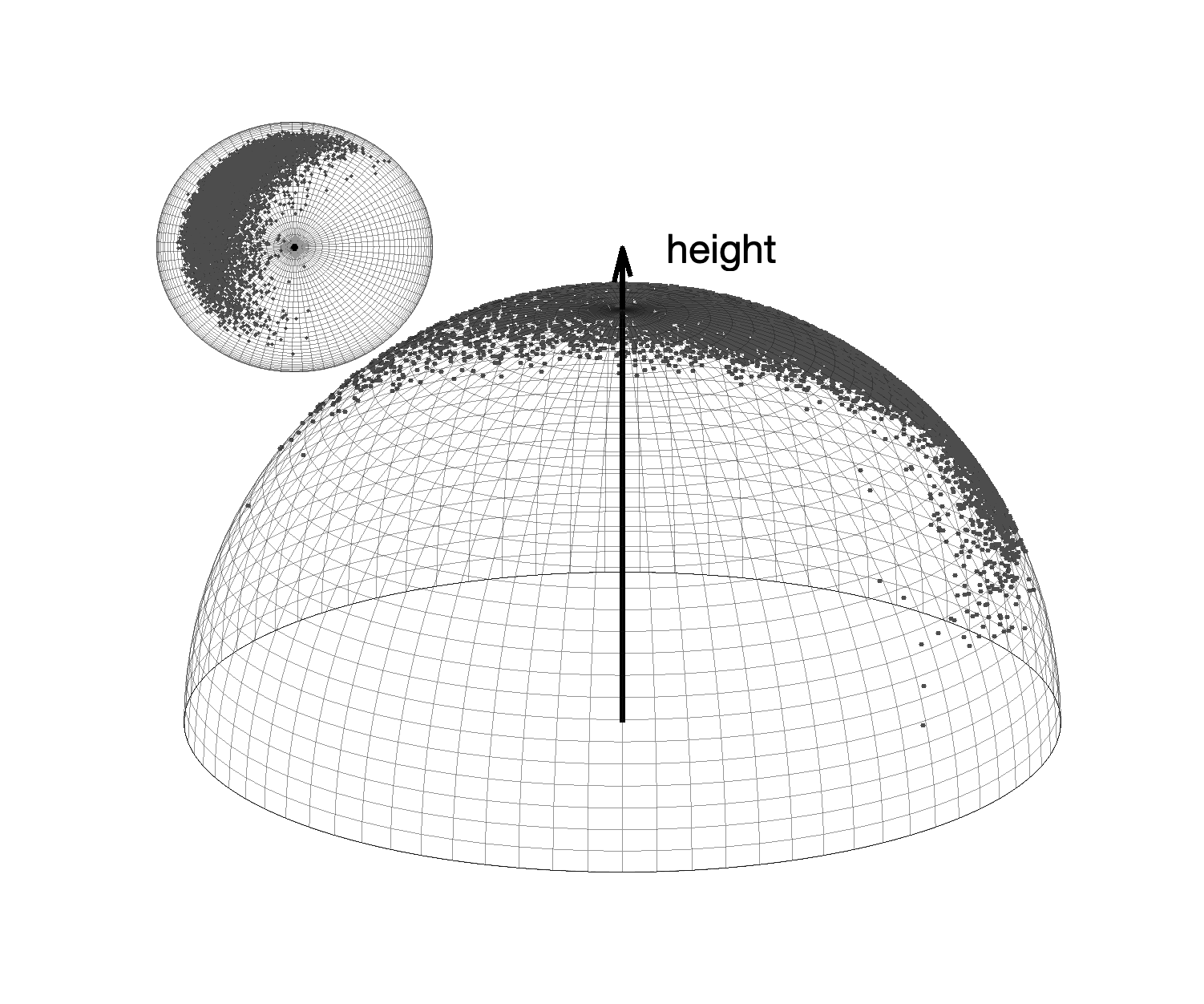}} \\
    \end{tabular}
\caption{Point clouds (longitude   $2\alpha$, latitude   $\arcsin(2\sigma_1  \sigma_2)$,   height   $\sigma_1  \sigma_2$) on semi-sphere, when    $M\in\{3,9,15\}$, $\epsilon\in\{2,10\}$,   $\delta = 10^{-5}$, and $\Delta_2f=1$}
    \label{tab:sphere-view}
\end{table*}

In Table~\ref{tab:sphere-view}, we visualize   10,000 such points   on hemisphere when $M\in\{3,9,15\}$, $\epsilon\in\{2,6,10\}$, $\delta = 10^{-5}$, and $\Delta_2f=1$. Table~\ref{tab:sphere-view} reveals three clear geometric trends. First, for a fixed dimension $M$, increasing $\epsilon$ (i.e., injecting light noise) makes the point clouds move upward on the hemisphere, meaning that the typical height $h=\sigma_1\sigma_2$ becomes larger. This agrees with the theorem, since a larger $\epsilon$ corresponds to a smaller Gaussian noise scale $\sigma$, and therefore the   signal component $\boldsymbol{v}$ plays a stronger role in the normalized triangle geometry. Second, for a fixed privacy level, increasing the dimension $M$ makes the clouds substantially more concentrated. Geometrically, this is because the magnitude of orthogonal  noise  component $n_\perp\sim \sigma\chi_{M-1}$ becomes more sharply concentrated around its typical value as $M$ grows, which suppresses low-dimensional variability and produces a thinner and more stable band on the hemisphere. Third, the point clouds are not rotationally symmetric in longitude: instead, they exhibit a clear directional bias and form crescent-like or band-like structures. This asymmetry comes from the   shift induced by the sensitivity vector $\boldsymbol{v}$, which breaks uniformity in the longitude direction   after Frobenius normalization. Taken together, these plots show that stronger privacy noise pushes the normalized random triangles toward lower latitudes, while higher dimension sharpens the concentration pattern and makes the hemisphere geometry increasingly structured.

 These plots show that stronger privacy, i.e., heavy noise, pushes the normalized random triangles toward lower latitudes, while higher dimension sharpens the concentration pattern and makes the hemisphere geometry increasingly structured. In particular, we have the following propositions which are proved in  Appendix~\ref{app:prove-bijective} and
Appendix~\ref{app:diffeomorphism}.

\begin{prop}\label{prop:bijective-mapping}
   Define the Gram matrix associated with $\widetilde{\Delta}\mathbf{H}$ as 
   \begin{equation*}
       \mathbf{G}
    =
    (\widetilde{\Delta}\mathbf{H})^\top
    (\widetilde{\Delta}\mathbf{H})
    =
    \mathbf{H}^\top
    \widetilde{\Delta}^\top
    \widetilde{\Delta}
     \mathbf{H}\in\R^{2\times 2},
   \end{equation*}
 which has   eigendecomposition  
\begin{equation*}
      \mathbf{G}
    =
\mathbf{V}
    \begin{bmatrix}
        \sigma_1^2 & 0\\\\
        0 & \sigma_2^2
    \end{bmatrix}
   \mathbf{V}^\top.
\end{equation*}
Then the algebraic bijection mapping between the SVD components ($\sigma_1$, $\sigma_2$) of $\widetilde{\Delta}\mathbf{H}$ and the 
coordinates (longitude, latitude) are
\begin{equation*}
\begin{aligned}
    \mathbf{G}_{11} - \mathbf{G}_{22} &= (\sigma_1^2 - \sigma_2^2)(\cos^2\alpha - \sin^2\alpha) = \cos\phi \cos\theta,\\
    2\mathbf{G}_{12} &= (\sigma_1^2 - \sigma_2^2)(2\sin\alpha \cos\alpha) = \cos\phi \sin\theta. 
    \end{aligned}
\end{equation*}
\end{prop}

\begin{prop}\label{prop:diffeomorphism}
    Let $\mathbf{n}_\parallel$ and $\mathbf{n}_\perp$ be the parallel and orthogonal components of $\mathbf{n}$ with respect to the sensitive vector $\boldsymbol{v}$ (as shown in Figure~\ref{fig:geo-dp}(a)), and define scalar variables   ${n}_\parallel$ and ${n}_\perp$ as
    \begin{equation*}
            n_\parallel = \frac{\langle \mathbf{n}, \bm{v} \rangle}{\|\bm{v}\|} \in \mathbb{R} \quad \text{and}\quad 
            n_\perp = \left\|\mathbf{n} - n_\parallel \frac{\bm{v}}{\|\bm{v}\|}\right\|_2 \in \mathbb{R}^+.
    \end{equation*}
Then the   diffeomorphism mappings     $ (n_\parallel, n_\perp) \mapsto (\tan\theta, \sin\phi) $ and $(\theta, \phi) \mapsto (n_\parallel, n_\perp)$ hold
\begin{equation}\label{eq:diffeomorphism}
    \left\{
    \begin{aligned}
    \tan \theta &= \frac{(\Delta_2f) n_\perp + 2 n_\parallel n_\perp}{(\Delta_2f)^2 + (\Delta_2f) n_\parallel + n_\parallel^2 - n_\perp^2}, \\
\sin\phi &= \frac{\sqrt{3} (\Delta_2f)n_\perp}{(\Delta_2f)^2 + (\Delta_2f)n_\parallel + n_\parallel^2 + n_\perp^2}.
\end{aligned}
\right.\quad \text{and} \quad
    \left\{
    \begin{aligned}
    n_\parallel &= \frac{\sqrt{3}\Delta_2f}{2} \cot\phi \sin\theta - \frac{\Delta_2f}{2}, \\
    n_\perp &= \frac{\sqrt{3}\Delta_2f}{2} \frac{1 - \cos\phi \cos\theta}{\sin\phi}.
\end{aligned}
\right.
\end{equation}
\end{prop}
\begin{remark}
The longitude $\theta$ is in the range of $[0,2\pi)$. By checking the signs of $\sin\phi$, $\cos\phi\cos\theta$ in \eqref{eq:alg_cos_theta}, and $\cos\phi \sin\theta$ in \eqref{eq:alg_sin_theta}, one can      uniquely determine the value of $\theta$ from $\tan\theta$. 
\end{remark}

By leveraging the conclusions in Proposition~\ref{prop:bijective-mapping} and~\ref{prop:diffeomorphism}, we prove the PDF of $(\theta,\phi)$. The details are deferred to 
Appendix~\ref{app:proof-spherepdf}.
\begin{theorem}\label{thm:spherepdf}
    Let $M$ be the dimension of the sensitive query $f(\bm{x})$, $\Delta_2f$ is the $l_2$ sensitivity, and i.i.d. Gaussian noise with variance $\sigma^2$ is used to obscure $f(\bm{x})\in\R^M$. Define the position (longitude $\theta$ and latitude $\phi$) of the normalized random triangles  as in \eqref{eq:long-lat}, then the joint PDF of $(\theta,\phi)$ is
    \begin{equation*}
    f(\theta,\phi) = \frac{1}{ \sqrt{2\pi} 2^{\frac{M-1}{2}-1}\Gamma\left(\frac{M-1}{2}\right)} \left(\frac{\sqrt{3}\Delta_2f}{2\sigma}\right)^M  \frac{\cos\phi (1 - \cos\phi \cos\theta)^{M-1}}{\sin^{M+1}\phi}   \exp\left( - \frac{(\Delta_2f)^2}{8\sigma^2} \Phi(\theta, \phi) \right),
\end{equation*}
where $\Phi(\theta, \phi)$ is  given by 
\begin{equation*}
    \Phi(\theta, \phi) = \frac{3(1 + \cos^2\phi - 2\cos\phi\cos\theta) - 2\sqrt{3}\cos\phi\sin\phi\sin\theta}{\sin^2\phi} + 1.
\end{equation*}
\end{theorem}


Similar to the first geometric tool, we also provide a corollary (proved in Appendix~\ref{app:sphere-plrv-proof}) that connects the traditional PLRV to a PLRV supported on the hemisphere. 
\begin{corollary}\label{corollary:sphere-plrv} The PLRV associated with the classic Gaussian mechanism in (\ref{eq:plrv_reduction}) can be reconstructed by a random longitude and latitude   $(\theta,\phi)$ generated by the PDF in Theorem~\ref{thm:spherepdf} via
\begin{equation*}
    \text{PLRV}(\theta, \phi) 
    = -\frac{\sqrt{3}(\Delta_2f)^2}{2\sigma^2} \cot\phi \sin\theta.
\end{equation*}
\end{corollary}

\begin{remark}
Proposition~\ref{corollary:sphere-plrv} reveals that privacy loss is  purely governed by      sine of the longitude and the     cotangent of the latitude. Similar to Corollary~\ref{corollary:new-plrv}, $\text{PLRV}(\theta, \phi)$ is   another reparameterization of the classic PLRV in \eqref{eq:plrv_reduction}. Thus, the two PLRVs should exhibit  the same tail probability bound. In other words, performing measure concentration analysis on $f(\theta,\phi)$ will lead to the same DP guarantee with the Gaussian mechanism. 
\end{remark}


The hemisphere representation provides the spectral counterpart of the
simplex collapse established in Section~\ref{sec:task1.1}. Whereas the simplex
coordinates describe the relative squared edge lengths of the normalized
triangle, the latitude on the hemisphere describes the balance between
its two nonzero singular values. Proposition~\ref{prop:hemisphere-force}, proved in Appendix~\ref{app:hemisphere-force}, characterizes the asymptotic
behavior of this spectral shape as the ambient dimension increases.

\begin{prop}[Equatorial Drift and  Band Concentration]\label{prop:hemisphere-force}
 Consider the hemisphere representation of the Gaussian mechanism, with spherical coordinates \((\theta,\phi)\) defined by \eqref{eq:long-lat} and joint density \(f(\theta,\phi)\) given by Theorem~\ref{thm:spherepdf}. As the dimension \(M\) increases, the random point cloud on the hemisphere exhibits the following two asymptotic effects:
\begin{enumerate}
    \item \textbf{Equatorial drift.}
    The dominant mass moves toward lower latitudes, i.e., toward the equator, at polynomial rate \(M^{-1/2}\). 
    \item \textbf{Band concentration.}
    At the same time, the point cloud becomes increasingly compressed into a thin band around the equator.
\end{enumerate}   
\end{prop}

\begin{remark}
The proposition distinguishes the location and width of the
high-dimensional hemisphere distribution. Equatorial drift describes the
decrease of its typical latitude, while band concentration describes the
shrinking variation around that typical latitude. Geometrically, the
normalized triangles approach the rank-one boundary of the shape space, i.e., 
the second singular value becomes small relative to the first, and hence $\sigma_1\sigma_2
    =
    \frac12\sin\phi
    \longrightarrow0$.
\end{remark}

\noindent\textbf{Geometric interpretation.} The simplex and hemisphere limits describe the same high-dimensional
degeneration in two complementary coordinate systems. In the simplex
representation, 
    $(x,y,z)
    \longrightarrow
    \left(0,\frac12,\frac12\right)$, 
which means that the sensitivity edge becomes negligible relative to the
two noise-related edges after Frobenius normalization. In the hemisphere
representation, 
    $\phi\longrightarrow0$, 
which means that the normalized triangle approaches the rank-one boundary
of spectral shape space. Although the scalar PLRV can be reconstructed from either
representation, the two geometric distributions describe different
aspects of the random triangle shapes associated with the Gaussian
mechanism.

\section{Conclusion}\label{sec:conclusion}

We studied the classical Gaussian mechanism through the random triangles
induced by its perturbation instances. We introduced two complementary
geometric representations. The simplex representation characterizes the
distribution of normalized squared edge lengths, while the hemisphere
representation characterizes spectral shape through the two nonzero singular
values. For both representations, we derived
the induced probability densities and showed how the classical scalar PLRV
can be reconstructed from the corresponding geometric coordinates. Our
high-dimensional analysis further establishes a metric collapse toward
$(0,\frac12,\frac12)$ on the simplex and a corresponding equatorial drift
and band concentration on the hemisphere.

These results establish a probabilistic bridge between differential privacy
and random shape theory, providing two exact geometric coordinate systems for
the classical Gaussian mechanism. More broadly,
they connect DP with random shape analysis and illustrate
how a scalar PLRV can be complemented by distributions
over geometric objects. Natural directions for future work include extending
the framework to anisotropic Gaussian and non-Gaussian perturbations,
incorporating the geometry of sensitivity sets, and investigating whether
geometric representations can support mechanism comparison, privacy--utility
diagnostics, or geometry-aware noise design.

\newpage
\bibliographystyle{unsrtnat}
\bibliography{sample.bib}
\newpage

\appendix

\section{Omitted Proofs}\label{app:proofs}

\subsection{Proof of Theorem~\ref{thm:simplex-density}}\label{app:simplex-proof}

As visualized in Figure~\ref{fig:geo-dp},
we first decompose the noise vector $\mathbf{n}$ into two parts, one is parallel to $\boldsymbol{v}$ and the other is orthogonal to $\boldsymbol{v}$, i.e., 
\begin{equation*}
   \mathbf{n} = a \frac{\boldsymbol{v}}{\|\boldsymbol{v}\|}+\boldsymbol{\omega}. 
\end{equation*}
 Denote $\|\boldsymbol{\omega}\|_2^2 \triangleq r^2$, we have $a\sim\mathcal{N}(0,\sigma^2)$, $r^2\sim \sigma^2 \chi^2_{M-1}$ (a chi squared random variable with $M-1$ degree of freedom scaled by $\sigma^2$), and $a \bot r^2$ ($\bot$ indicates independence). Then, we have 
\begin{equation}\label{eq:xyzS}
\left\{
\begin{aligned}
\|\boldsymbol{e}_1\|^2+\|\boldsymbol{e}_2\|^2+\|\boldsymbol{e}_3\|^2 
\triangleq S  &
=  \|\boldsymbol{v}\|^2 +\|\mathbf{n}\|^2 + \|\mathbf{n}+\boldsymbol{v}\|^2 \\
& = 2\left(a^2+(\Delta_2 f)^2+a\Delta_2 f+r^2\right),\\
\|\boldsymbol{e}_1\|^2 \triangleq x &
= \frac{\|\boldsymbol{v}\|^2}{S} 
= \frac{(\Delta_2 f)^2}{S}, \\
\|\boldsymbol{e}_2\|^2 \triangleq y &
= \frac{\|\mathbf{n}\|^2}{S} 
= \frac{a^2 + r^2}{S}, \\
\|\boldsymbol{e}_3\|^2 \triangleq z &
= \frac{\|\mathbf{n}+\boldsymbol{v}\|^2}{S} \\
&= \frac{a^2+(\Delta_2 f)^2+2a\Delta_2 f+r^2}{S}\\
& = 1-x-y
\end{aligned}
\right.
\end{equation}
For a given tuple  $(x,y)$, we solve for $a$ and $r^2$ and have the following inverting mapping
    \begin{equation*}
\left\{
\begin{aligned}
a \triangleq p(x,y) = & \frac{\Delta_2f (1-2x-2y)}{2x}\\
r^2 \triangleq q(x,y) = & \frac{\left(\Delta_2f\right)^2 y}{x}-\big(p(x,y)\big)^2
\end{aligned}.
\right.
\end{equation*}
Then, we can  obtain the joint PDF of $(x,y)$ by change of variables, i.e., 
$$f_{X,Y}(x,y)
= f_{A,R^2}\!\big(p(x,y),q(x,y)\big)\,
\left|\det\frac{\partial(p,q)}{\partial(x,y)}\right|,$$ where the 
Jacobian determinant is 
\begin{equation*}
J(x,y)
=
\det
\begin{bmatrix}
\partial p/\partial x & \partial p/\partial y\\[2pt]
\partial q/\partial x & \partial q/\partial y
\end{bmatrix}
= p_x q_y - p_y q_x.
\end{equation*}
One can verify that 
\begin{align*}
p_x &= \frac{\Delta_2 f}{2}\cdot\frac{2y-1}{x^{2}}, &
p_y &= -\,\frac{\Delta_2 f}{x}.
\end{align*}
Since $q(x,y)=\tfrac{(\Delta_2 f)^2 y}{x} - p^2$, we have
\begin{align*}
q_x &= -\,\frac{(\Delta_2 f)^2 y}{x^{2}} - 2p\,p_x, &
q_y &=  \frac{(\Delta_2 f)^2}{x} - 2p\,p_y.
\end{align*}
It is easy to check that $\left|\det\frac{\partial(p,q)}{\partial(x,y)}\right|
= \frac{(\Delta_2 f)^3}{2x^{3}}$. Hence, using the fact that $A\bot R^2$,   the joint density of $(x,y)$ is
\begin{equation*}
\begin{aligned}
f_{X,Y}(x,y)
&= f_{A}\big(p(x,y)\big)\times f_{R^2}\big(q(x,y)\big)\,
\frac{(\Delta_2 f)^3}{2x^{3}}\\
&=\frac{(\Delta_2 f)^3}{2x^3}
\cdot
\frac{1}{\sqrt{2\pi}\,\sigma}
\exp\!\left[-\frac{1}{2\sigma^2}p(x,y)^2\right]
\\[6pt]
&\quad\times
\frac{\big(q(x,y)\big)^{\frac{M-3}{2}}}
{2^{\frac{M-1}{2}}
\sigma^{M-1}
\Gamma\!\left(\frac{M-1}{2}\right)}
\exp\!\left[-\frac{1}{2\sigma^2}q(x,y)\right] \\
&\text{for all } x>0,\;y>0,\;x+y<1,\;q(x,y)\ge0
\end{aligned},
\end{equation*}
where $p(x,y)$ and $q(x,y)$ are given in (\ref{eq:p-q}). This completes the proof.

\subsection{Proof of Corollary  \ref{corollary:new-plrv}}\label{proof-simplex-plrv}
\begin{proof} First, by leveraging $x,y,z$, and $S$ defined in (\ref{eq:xyzS}), 
we have
\begin{equation*}
    \begin{aligned}
\langle \bm{v},\mathbf{n}\rangle   = &\frac{S}{2} \frac{\|\bm{v}+\mathbf{n}\|^2 - \|\bm{v}\|^2 -\|\mathbf{n}\|^2 }{S}     \\
= & \frac{S}{2} (z-x-y)\\
= & \frac{S}{2} \bigl((1-x-y)-x-y\bigr)\\
= &  \frac{1}{2}\frac{\left(\Delta_2f\right)^2}{x} (1-2x-2y)
    \end{aligned}
\end{equation*}
Then by plugging the above equation into (\ref{eq:plrv_reduction}), we arrive at
\begin{equation*}
\begin{aligned}
\mathrm{PLRV} =  & -\frac{1}{2\sigma^2} \left(\|\bm{v}\|^2+2\langle \mathbf{n},\bm{v}\rangle\right)\\ =&  -\frac{1}{2\sigma^2} \left(\left(\Delta_2f\right)^2+\frac{\left(\Delta_2f\right)^2}{x} (1-2x-2y)\right)\\
=& -\frac{\left(\Delta_2f\right)^2}{2\sigma^2}  \left(1+\frac{1}{x} (1-2x-2y)\right) \\= & \frac{\left(\Delta_2f\right)^2}{2\sigma^2} \cdot \frac{x+2y-1}{x} \triangleq L(x,y),
\end{aligned}
\end{equation*}
which concludes the proof.
\end{proof}

\subsection{Proof of Proposition~\ref{prop:phase_transition}}\label{app:proof-phase_transition}

\begin{proof}
The conclusion follows directly from the structure of the joint density
\eqref{eq:jointpdf}. Up to a normalization constant,
\[
f_{X,Y}(x,y)
\propto
x^{-3}
\exp\!\left(-\frac{p(x,y)^2}{2\sigma^2}\right)
q(x,y)^{\frac{M-3}{2}}
\exp\!\left(-\frac{q(x,y)}{2\sigma^2}\right).
\]
Among these factors, the only explicit $M$-dependence appears in the power $q(x,y)^{\frac{M-3}{2}}$.
Yet, the dependence of the density on the scalar variable $q(x,y)$ is governed by the full factor $q(x,y)^{\frac{M-3}{2}}
\exp\!\left(-\frac{q(x,y)}{2\sigma^2}\right)$, which is the kernel of a  Gamma distribution.
 Hence, for
large $M$, the dominant contribution to the density comes from the region where
$q(x,y)$ takes its typical high-dimensional scale, i.e., $q(x,y)=\Theta(M)$ (as shown in Appendix~\ref{app:simplex-proof}, $q(x,y)\sim\chi^2_{M-1}$). 

In contrast, since $p(x,y)\sim\mathcal{N}(0,\sigma^2)$ which is independent of $M$ (also shown in Appendix~\ref{app:simplex-proof}), 
 the dominant mass of $p(x,y)$
must remain in the regime where $p(x,y)=O(1)$. Then, according to \eqref{eq:p-q}, it implies that $1-2x-2y = O(x)$, or equivalently,
$y=\frac12+O(x)$.

Substituting this relation into $q(x,y)$   
and using $q(x,y)=\Theta(M)$ together with $p(x,y)=O(1)$, we obtain $\frac{(\Delta_2 f)^2 y}{x}=\Theta(M)$. 
Since $y$ stays bounded away from $0$ and converges to $1/2$, this yields $x=\Theta\!\left(\frac{1}{M}\right)$. Consequently, $y=\frac12+O\!\left(\frac{1}{M}\right)$. Because $z=1-x-y$, we also have $z=\frac12+O\!\left(\frac{1}{M}\right)$.

Additionally, recall \eqref{eq:p-q}
    \begin{equation*}
\left\{
\begin{aligned}
a \triangleq p(x,y) = & \frac{\Delta_2f (1-2x-2y)}{2x}\\
r^2 \triangleq q(x,y) = & \frac{\left(\Delta_2f\right)^2 y}{x}-\big(p(x,y)\big)^2
\end{aligned},
\right.
\end{equation*}
and use the fact $r^2>0$, we have $\frac{\left(\Delta_2f\right)^2 y}{x}> \left( \frac{\Delta_2f (1-2x-2y)}{2x}\right)^2$, which implies   
\begin{equation*}
    \left(x-\frac13\right)^2
    +
    \left(x-\frac13\right)
    \left(y-\frac13\right)
    +
    \left(y-\frac13\right)^2
    <
    \frac1{12},
\end{equation*}
 a rotated ellipse centered at $(1/3,1/3)$.
\end{proof}






\subsection{Proof of Proposition~\ref{prop:bijective-mapping}}
\label{app:prove-bijective}
\begin{proof}
Given \eqref{eq:unique-svd}, we have the   (normalized) Gram matrix $\mathbf{G}$ characterizing the internal geometry of the random triangle as 
\begin{align}\label{eq:gram}
    \mathbf{G} &=(\widetilde{\Delta}\mathbf{H})^\top
    (\widetilde{\Delta}\mathbf{H})
    =
    \mathbf{H}^\top
    \widetilde{\Delta}^\top
    \widetilde{\Delta}
     \mathbf{H}\nonumber\\
    &= \begin{bmatrix} \cos\alpha & -\sin\alpha  \\ \sin\alpha & \cos\alpha \end{bmatrix} \begin{bmatrix} \sigma_1^2 & 0  \\ 0 & \sigma_2^2 \end{bmatrix} \begin{bmatrix} \cos\alpha & \sin\alpha    \\ -\sin\alpha & \cos\alpha   \end{bmatrix} \nonumber \\
    &= \begin{bmatrix} \sigma_1^2 \cos^2\alpha + \sigma_2^2 \sin^2\alpha & (\sigma_1^2 - \sigma_2^2) \sin\alpha \cos\alpha   \\ (\sigma_1^2 - \sigma_2^2) \sin\alpha \cos\alpha & \sigma_1^2 \sin^2\alpha + \sigma_2^2 \cos^2\alpha  \end{bmatrix}.
\end{align}
Because the triangle is strictly normalized by the Frobenius norm (i.e., $\|\widetilde{\triangle}\|_F = 1$), the trace of $\mathbf{G}$ is constrained to $\Tr(\mathbf{G}) = \sigma_1^2 + \sigma_2^2 = 1$.  Note again, $\widetilde{\Delta}\mathbf{H}$ and $\widetilde{\Delta}$ have the same two nonzero
singular values since  $(\widetilde{\Delta}\mathbf{H})(\widetilde{\Delta}\mathbf{H})^\top = \widetilde{\Delta}\widetilde{\Delta}^\top$.

Next, we associate the SVD parameters with the coordinates on the hemisphere of radius $1/2$. Using basic trigonometric identities, we rewrite   elements of the Gram matrix $\mathbf{G}$ strictly in terms of the longitude $\theta$ and latitude $\phi$. First, according to \eqref{eq:long-lat}, $2\sigma_1\sigma_2 = \sin \phi$. Then use
$\sigma_1\geq \sigma_2\geq 0$ and $\sigma_1^2+\sigma_2^2=1$, we have
\begin{equation}\label{eq:cosphi}
    \sigma_1^2 - \sigma_2^2 = \sqrt{(\sigma_1^2 + \sigma_2^2)^2 - 4\sigma_1^2 \sigma_2^2} = \sqrt{1 - \sin^2\phi} = \cos\phi.
\end{equation}
By leveraging \eqref{eq:cosphi} and double-angle formulas for cosine and sine, i.e., 
\begin{equation}\label{eq:double-angular}
    \theta = 2\alpha, \quad \cos \theta = \cos^2\alpha - \sin^2\alpha, \quad \text{and}\ \sin \theta = 2\sin \alpha \cos \alpha,
\end{equation}
the matrix elements of $\mathbf{G}$ in \eqref{eq:gram} map perfectly to the spherical coordinates. In particular, 
\begin{align}
    \mathbf{G}_{11} - \mathbf{G}_{22} &= (\sigma_1^2 - \sigma_2^2)(\cos^2\alpha - \sin^2\alpha) = \cos\phi \cos\theta, \label{eq:coord_1} \\
    2\mathbf{G}_{12} &= (\sigma_1^2 - \sigma_2^2)(2\sin\alpha \cos\alpha) = \cos\phi \sin\theta. \label{eq:coord_2}
\end{align}
Equations \eqref{eq:coord_1} and \eqref{eq:coord_2} establish a  algebraic bijection between the SVD components and the spherical coordinates $(\theta, \phi)$.
\end{proof}

\subsection{Proof of Proposition~\ref{prop:diffeomorphism}}\label{app:diffeomorphism}
Notice that the triangle defined by the vertices $\{\bm{v}, \mathbf{n}, -\bm{v}-\mathbf{n}\}$ lies strictly within the 2D subspace spanned by the   signal $\bm{v}$ and the random noise $\mathbf{n}$. Therefore, the entire intrinsic geometry of the triangle is governed solely by two scalar components of $\mathbf{n}$ projected onto this   2D plane. The first component is  the scalar projection parallel to the signal, i.e.,  
\begin{equation}\label{eq:firstpart} 
n_\parallel = \langle \mathbf{n}, \bm{v} \rangle / \|\bm{v}\| \in \mathbb{R},
\end{equation}
and the second component is the magnitude of the orthogonal rejection, i.e., 
\begin{equation}\label{eq:secondpart} 
n_\perp = \left\|\mathbf{n} - n_\parallel \frac{\bm{v}}{\|\bm{v}\|}\right\|_2 \in \mathbb{R}^+.
\end{equation}
Let $\Delta_2f= \|\bm{v}\|$. In the local 2D coordinates, 
the three vertices of the random triangle can be equivalently represented as Cartesian points 
\begin{equation}
    \bm{y}_1 = \begin{bmatrix} \Delta_2f\\ 0 \end{bmatrix}, \quad 
    \bm{y}_2 = \begin{bmatrix} n_\parallel \\ n_\perp \end{bmatrix}, \quad \text{and} \quad 
    \bm{y}_3 = \begin{bmatrix} -\Delta_2f-n_\parallel \\ -n_\perp \end{bmatrix}.
\end{equation}
We aggregate these column vectors into a $2 \times 3$ raw coordinate matrix $\mathbf{Y} = \begin{bmatrix} \bm{y}_1 & \bm{y}_2 & \bm{y}_3 \end{bmatrix} = \begin{bmatrix} \Delta_2f& n_\parallel & -\Delta_2f-n_\parallel \\ 0 & n_\perp & -n_\perp \end{bmatrix}$. Notice that $\mathbf{Y} \mathbf{1}  = \mathbf{0}$, meaning the centroid of the 2D random triangle is also intrinsically anchored at the origin $(0,0)^T$.

Next, we analyze the unnormalized Gram matrix $\mathbf{G}_{un}\in\R^{2\times 2}$ which encodes all internal geometric features (edge lengths, angles, and area) and is defined as 
\begin{equation}
\begin{aligned}
        \mathbf{G}_{un} \triangleq \mathbf{Y} \mathbf{Y}^T &= \begin{bmatrix} \Delta_2f& n_\parallel & -\Delta_2f-n_\parallel \\ 0 & n_\perp & -n_\perp \end{bmatrix} \begin{bmatrix} \Delta_2f & 0 \\ n_\parallel & n_\perp \\ -\Delta_2f-n_\parallel & -n_\perp \end{bmatrix} \\
        &= \begin{bmatrix} 2\left((\Delta_2f)^2 + (\Delta_2f)n_\parallel + n_\parallel^2\right) & n_\perp(\Delta_2f+ 2n_\parallel) \\ n_\perp(\Delta_2f+ 2n_\parallel)& 2n_\perp^2 \end{bmatrix}.
\end{aligned}
\end{equation}
Observing that the squared Frobenius norm of the unnormalized triangle satisfies 
\begin{equation*}
\begin{aligned}
    \|\mathbf{Y}\|_F^2 &=
      \|\bm{y}_1\|_2^2+\|\bm{y}_2\|_2^2+\|\bm{y}_3\|_2^2\\&=2((\Delta_2f)^2 + \Delta_2fn_\parallel + n_\parallel^2 + n_\perp^2)=
    \Tr(\mathbf{G}_{un}). 
    \end{aligned}
\end{equation*}

Clearly,  the random triangle is normalized by its Frobenius norm to isolate its pure intrinsic shape, the corresponding normalized Gramm matrix (defined in Proposition~\ref{prop:bijective-mapping}) is exactly $\mathbf{G} = \mathbf{G}_{un} / \text{Tr}(\mathbf{G}_{un})$. This rigorous algebraic link allows us to uniquely connect the   SVD singular values to the physical noise projections $n_\parallel$ and $n_\perp$.

Since the normalized Gram matrix in \eqref{eq:gram} is connected to the unnormalized one via 
\begin{equation}\label{eq:gram-ungram}
    \mathbf{G} = \begin{bmatrix}
          \mathbf{G}_{un} / \|\mathbf{Y}\|_F^2 & 0 \\
          0 &0 
    \end{bmatrix}=\begin{bmatrix}
         \mathbf{G}_{un} /\Tr(\mathbf{G}_{un}) & 0 \\
          0 &0  
    \end{bmatrix},
\end{equation}
we have $\mathbf{G}_{11}\mathbf{G}_{22}-\mathbf{G}_{12}\mathbf{G}_{21} = \det\left(\frac{\mathbf{G}_{un}}{\|\mathbf{Y}\|_F^2}\right)$, which suggests $(\sigma_1\sigma_2)^2 = \det\left(\frac{\mathbf{G}_{un}}{\|\mathbf{Y}\|_F^2}\right)$. By leveraging the expression of the latitude in \eqref{eq:long-lat}, we have
\begin{equation}\label{eq:sinphi}
\begin{aligned}
\sin\phi &= 2\sigma_1\sigma_2\\
    &     = 2\sqrt{\det\left(\frac{\mathbf{G}_{un}}{\|\mathbf{Y}\|_F^2}\right)} = 2 \frac{\sqrt{\det(\mathbf{G}_{un})}}{\|\mathbf{Y}\|_F^2} = \frac{\sqrt{3} (\Delta_2f)n_\perp}{\left((\Delta_2f)^2 + (\Delta_2f)n_\parallel + n_\parallel^2 + n_\perp^2\right)}.
\end{aligned}
\end{equation}

Similarly, by leveraging  \eqref{eq:coord_1} and \eqref{eq:coord_2}, we have
\begin{align}
    \cos\phi \cos\theta &= \mathbf{G}_{11} - \mathbf{G}_{22} = \frac{{\mathbf{G}_{un}}_{11}-{\mathbf{G}_{un}}_{22}}{\|\mathbf{Y}\|_F^2}= \frac{(\Delta_2f)^2 + (\Delta_2f) n_\parallel + n_\parallel^2 - n_\perp^2}{(\Delta_2f)^2 + (\Delta_2f) n_\parallel + n_\parallel^2 + n_\perp^2},
\label{eq:alg_cos_theta} \\
    \cos\phi \sin\theta &=   2\mathbf{G}_{12}= \frac{2{\mathbf{G}_{un}}_{12}}{\|\mathbf{Y}\|_F^2} = \frac{(\Delta_2f) n_\perp + 2 n_\parallel n_\perp}{(\Delta_2f)^2 + (\Delta_2f) n_\parallel + n_\parallel^2 + n_\perp^2}. \label{eq:alg_sin_theta}    
\end{align} 
Dividing \eqref{eq:alg_sin_theta} by \eqref{eq:alg_cos_theta} yields $\tan\theta$.

Combine \eqref{eq:sinphi}, \eqref{eq:alg_cos_theta}, and \eqref{eq:alg_sin_theta}, one can verify that \eqref{eq:firstpart} and \eqref{eq:secondpart} can be parameterized as
\begin{align*}
    n_\parallel(\theta, \phi) &= \frac{\sqrt{3}\Delta_2f}{2} \cot\phi \sin\theta - \frac{\Delta_2f}{2}, \\
    n_\perp(\theta, \phi) &= \frac{\sqrt{3}\Delta_2f}{2} \frac{1 - \cos\phi \cos\theta}{\sin\phi}.
\end{align*}
Note that 
$n_\perp(\theta, \phi)$ is obtained by solving  a quadratic function w.r.t. $n_{\perp}$, i.e., 
\begin{equation*}
    n_{\perp}^2 + \frac{\sqrt{3}(\Delta_2f)\cos \phi \cos \theta}{\sin \phi}n_{\perp}-\frac{3(\Delta_2f)^2}{4}(1+\cot^2\phi\sin^2\theta) = 0,
\end{equation*}
and we only select the   physically meaningful root.

\subsection{Proof of Theorem~\ref{thm:spherepdf}}\label{app:proof-spherepdf}

To derive the joint PDF of the longitude and latitude, $(\theta, \phi)$, we must trace the spherical coordinates back to the original $M$-dimensional noise vector $\mathbf{n} \sim \mathcal{N}(\mathbf{0}, \sigma^2 \mathbf{I}_M)$.  Since the $M$-dimensional noise $\mathbf{n}$ is isotropic Gaussian, its projected components are independent, with $n_\parallel \sim \mathcal{N}(0, \sigma^2)$ and $n_\perp \sim \sigma\chi_{M-1}$. The joint PDF in the noise domain is inherently proportional to the product of $f(n_\parallel)$ and $f(n_\perp)$, hence
\begin{equation}
\begin{aligned}
      f(\mathbf{n}) & = f_{\text{noise}}(n_\parallel, n_\perp) =f(n_\parallel)\cdot f( n_\perp)\\
      &=\frac{1}{\sqrt{2\pi}\sigma}\exp\left(-\frac{n_\parallel^2}{2\sigma^2}\right)\cdot \frac{1}{2^{\frac{M-1}{2}-1}\Gamma\left(\frac{M-1}{2}\right)\sigma^{M-1}}n_\perp^{M-2} \exp\left(-\frac{n_\perp^2}{2\sigma^2}\right) \\&
      = \frac{1}{ \sqrt{2\pi}\sigma 2^{\frac{M-1}{2}-1}\Gamma\left(\frac{M-1}{2}\right)\sigma^{M-1}} n_\perp^{M-2} \exp\left(-\frac{n_\parallel^2 + n_\perp^2}{2\sigma^2}\right).
\end{aligned}
\end{equation}

Let $\Psi : (n_\parallel, n_\perp) \mapsto (\theta, \phi)$ denote the forward algebraic mapping to the tuple of the longitude and latitude. To evaluate the joint PDF of $(\theta,\phi)$, we apply  the change of variables theorem, i.e.,
\begin{equation}\label{eq:pdf-theta-phi}
\begin{aligned}
        f(\theta, \phi) &= f_{\text{noise}}\left(n_\parallel(\theta, \phi), n_\perp(\theta, \phi)\right) \cdot \left| \det \mathbf{J}_{\Psi^{-1}}\right|
\end{aligned}
\end{equation}
where $\mathbf{J}_{\Psi^{-1}} = \frac{\partial(n_\parallel, n_\perp)}{\partial(\theta, \phi)}$ is the Jacobian matrix of the inverse mapping. According to  Proposition~\ref{prop:diffeomorphism}, we have
\begin{equation*}
    \mathbf{J}_{\Psi^{-1}} = \frac{\partial(n_\parallel, n_\perp)}{\partial(\theta, \phi)}   = \begin{bmatrix}
      \frac{\partial n_\parallel}{\partial \theta} & \frac{\partial n_\parallel}{\partial \phi}\\
        \frac{\partial n_\perp}{\partial \theta} & \frac{\partial n_\perp}{\partial \phi} 
    \end{bmatrix} = \begin{bmatrix}
      \frac{\sqrt{3} \Delta_2f}{2}  \cot\phi \cos \theta & -\frac{\sqrt{3} \Delta_2f}{2}\frac{\sin \theta}{\sin^2 \phi}\\
      \frac{\sqrt{3} \Delta_2f}{2} \frac{\cos \phi \sin\theta}{\sin \phi} & \frac{\sqrt{3} \Delta_2f}{2}\frac{\cos\theta - \cos\phi}{\sin^2\phi}
    \end{bmatrix},
\end{equation*}
and hence
\begin{equation}\label{eq:J-inverse}
  |   \mathbf{J}_{\Psi^{-1}}|   = \frac{3(\Delta_2f)^2}{4} \frac{\cos\phi (1 - \cos\phi \cos\theta)}{\sin^3\phi}.
\end{equation}
Additionally, we have 
\begin{equation}\label{eq:noise-energy}
    n_\parallel^2 + n_\perp^2 = \frac{(\Delta_2f)^2}{4} \left[ 3 \frac{1 + \cos^2\phi - 2\cos\phi\cos\theta}{\sin^2\phi} - 2\sqrt{3} \frac{\cos\phi\sin\theta}{\sin\phi} + 1 \right].
\end{equation} 
Substitute \eqref{eq:diffeomorphism}, \eqref{eq:J-inverse}, and \eqref{eq:noise-energy} into \eqref{eq:pdf-theta-phi}, we get 
\begin{equation*}
    f(\theta,\phi) = \frac{1}{ \sqrt{2\pi}  2^{\frac{M-1}{2}-1}\Gamma\left(\frac{M-1}{2}\right) } \left(\frac{\sqrt{3}\Delta_2f}{2 \sigma}\right)^M  \frac{\cos\phi (1 - \cos\phi \cos\theta)^{M-1}}{\sin^{M+1}\phi} \cdot \exp\left( - \frac{(\Delta_2f)^2}{8\sigma^2} \Phi(\theta, \phi) \right),
\end{equation*}
where $\Phi(\theta, \phi)$ is strictly given by
\begin{equation*}
    \Phi(\theta, \phi) = \frac{3(1 + \cos^2\phi - 2\cos\phi\cos\theta) - 2\sqrt{3}\cos\phi\sin\phi\sin\theta}{\sin^2\phi} + 1.
\end{equation*}

\subsection{Proof of Corollary~\ref{corollary:sphere-plrv}}\label{app:sphere-plrv-proof}

\begin{proof}
  According to (\ref{eq:plrv_reduction}), we have $$\mathrm{PLRV} =   -\frac{1}{2\sigma^2} \left(\|\bm{v}\|^2+2\langle \mathbf{n},\bm{v}\rangle\right) = -\frac{1}{2\sigma^2} \left((\Delta_2f)^2 + 2 (\Delta_2f) n_{\parallel}\right).$$
By plugging \eqref{eq:firstpart} and the coordinate mapping \eqref{eq:diffeomorphism} into the equation above, one can verify that
$\text{PLRV}(\theta, \phi)
    = -\frac{\sqrt{3}(\Delta_2f)^2}{2\sigma^2} \cot\phi \sin\theta.$
\end{proof}

\subsection{Proof of Proposition~\ref{prop:hemisphere-force}}\label{app:hemisphere-force}

\begin{proof}
    Indeed, from
    \[
    \sin\phi
    =
    \frac{\sqrt{3}\,\Delta_2 f\,n_\perp}
    {(\Delta_2 f)^2+(\Delta_2 f)n_\parallel+n_\parallel^2+n_\perp^2},
    \]
    where \(n_\parallel\sim \mathcal N(0,\sigma^2)\) and \(n_\perp\sim \sigma\chi_{M-1}\), the typical size of \(n_\parallel\) remains \(O(1)\), whereas \(n_\perp=O(\sqrt{M})\). Consequently, the numerator grows as \(O(\sqrt{M})\), while the denominator grows as \(O(M)\), yielding
    \[
    \sin\phi = O(M^{-1/2}),
    \qquad
    \phi = O(M^{-1/2}),
    \qquad
    h=\sigma_1\sigma_2=O(M^{-1/2}).
    \]
    Therefore, the hemisphere point cloud is driven toward the equator at rate \(M^{-1/2}\).

    At the same time, the cloud is not merely shifted toward the equator, but also compressed into an increasingly thin equatorial band. Indeed, since
\[
n_\perp \sim \sigma\chi_{M-1},
\]
we have \(n_\perp=\sigma\sqrt{M}+O_p(1)\), while \(n_\parallel=O_p(1)\), where $O_p(\cdot)$ denotes order in probability. Substituting this into
\[
\sin\phi
=
\frac{\sqrt{3}\,\Delta_2 f\,n_\perp}
{(\Delta_2 f)^2+(\Delta_2 f)n_\parallel+n_\parallel^2+n_\perp^2},
\]
gives
\[
\sin\phi
=
\frac{\sqrt{3}\,\Delta_2 f}{\sigma\sqrt{M}}
\bigl(1+O_p(M^{-1/2})\bigr).
\]
Thus, the latitude is not only of order \(M^{-1/2}\), but also increasingly concentrated around this scale. Equivalently, the hemisphere height
\[
h=\sigma_1\sigma_2=\frac12\sin\phi
\]
is confined to an increasingly narrow \(M^{-1/2}\)-scale neighborhood of the equator.
\end{proof}

\end{document}